\newif\ifARX\ARXtrue
\ifARX
 \documentclass[11pt, reqno]{article}
 \usepackage[hmargin=1.0in,vmargin=1.0in]{geometry}
 \usepackage[nodayofweek]{datetime}
 \usepackage{amsthm}
\else
 \documentclass[final]{siamltex}
\fi 

\usepackage[english]{babel}
\usepackage{wrapfig}
\usepackage{dsfont}
\usepackage{cite}
\usepackage{amsmath,amsfonts,amssymb,epsfig,graphicx}
\usepackage{mathrsfs,url}
\usepackage{tkz-graph}
\usepackage{enumerate}
\usepackage{microtype}
\usepackage{IEEEtrantools}

\usepackage{setspace}
\DeclareMathOperator{\VCSP}{VCSP}
\DeclareMathOperator{\BLP}{BLP}

\newcommand{\fpol}{\fpolVECTOR}
\newcommand{\fpolVECTOR}{\ensuremath{\calO^{(m)}\to\Rnn}}

\newcommand{\fpolMinusOne}{\fpolMinusOneVECTOR}
\newcommand{\fpolMinusOneVECTOR}{\ensuremath{\calO^{(m-1)}\to\Rnn}}

\newcommand{\tuple}[1]{\ensuremath{\langle #1\rangle}}

\newcommand{\R}{\mbox{$\mathbb R$}}
\newcommand{\Q}{\mbox{$\mathbb Q$}}
\newcommand{\Qnn}{\mbox{$\mathbb Q_{\geq 0}$}}
\newcommand{\Qnnc}{\mbox{$\overline{\mathbb Q}_{\geq 0}$}}
\newcommand{\Qc}{\mbox{$\overline{\mathbb Q}$}}
\newcommand{\Rnn}{\mbox{$\mathbb R_{\geq 0}$}}
\newcommand{\Rnnc}{\mbox{$\overline{\mathbb R}_{\geq 0}$}}
\newcommand{\Rc}{\mbox{$\overline{\mathbb R}$}}

\newcommand{\mm}[2]{\langle #1, #2 \rangle}

\newcommand{\meet}{\wedge}
\newcommand{\join}{\vee}

\ifARX
  \DeclareMathOperator{\supp}{supp}
\else
\fi 

\long\def\ignore#1{}
\def\myps[#1]#2{\includegraphics[#1]{#2}}

\def\br(#1,#2){{\langle #1,#2 \rangle}}
\def\setZ[#1,#2]{{[ #1 .. #2 ]}}

\def\q={\quad=\quad}
\def\qq={\qquad=\qquad}

\def\calC{{\cal C}}

\def\calI{{\cal I}}

\def\calO{{\cal O}}

\def\psfile[#1]#2{}
\def\psfilehere[#1]#2{}
\def\unitvec(#1){{{\bf u}_{#1}}}

\DeclareMathOperator{\dom}{\textup{\texttt{dom}}}

\newcommand{\bg}{\mbox{$\bf g$}}
\newcommand{\bh}{\mbox{$\bf h$}}

\newcommand{\eqdef}{=} 

\newtheorem{theoremX}{Theorem}
\newtheorem{lemmaX}[theoremX]{Lemma}
\newtheorem{propositionX}[theoremX]{Proposition}
\newtheorem{corollaryX}[theoremX]{Corollary}

\newtheorem{definitionX}{Definition}

\newtheorem{theoremRESTATED}{Theorem}
\newtheorem{lemmaRESTATED}{Lemma}

\newtheorem{remarkX}{Remark}
\newtheorem{exampleX}{Example}
\newtheorem{exampleRESTATED}{Example}

\def\examplefont{} %% italic in example environment
\ifARX

\fi

\ifARX
 \theoremstyle{remark}
\fi

\ifARX 
  
\begin{document}
\fi

\title{The power of linear programming for general-valued CSPs\thanks{Part of this work
(by J. Thapper and S. \v{Z}ivn\'y) appeared in the \emph{Proceedings of the 53rd
Annual IEEE Symposium on Foundations of Computer Science} (FOCS), pp. 669--678,
2012~\cite{tz12:focs}. Part of this work (by V. Kolmogorov) appeared in the
\emph{Proceedings of the 40th International Colloquium on Automata, Languages
and Programming} (ICALP), pp. 625--636, 2013~\cite{Kolmogorov13:icalp}. 
Vladimir Kolmogorov is supported by the European Research Council under the European Unions Seventh Framework Programme
(FP7/2007-2013)/ERC grant agreement no 616160.
Stanislav \v{Z}ivn\'y is supported by a Royal Society University Research Fellowship.
}}

\ifARX 

\author{Vladimir Kolmogorov\\
Institute of Science and Technology, Austria \\ 
\texttt{vnk@ist.ac.at}
\and
Johan Thapper\\
Universit\'e Paris-Est, Marne-la-Vall\'ee, France\\
\texttt{thapper@u-pem.fr}
\and
Stanislav \v{Z}ivn\'{y}\\
University of Oxford, UK\\
\texttt{standa.zivny@cs.ox.ac.uk}
}

\date{\today}
\maketitle

\else

\author{
Vladimir Kolmogorov\thanks{Institute of Science and Technology Austria, Austria (vnk@ist.ac.at)}
\and
Johan Thapper\thanks{Universit\'e Paris-Est, Marne-la-Vall\'ee, France (thapper@u-pem.fr)}
\and 
Stanislav \v{Z}ivn\'{y}\thanks{University of Oxford, UK (standa@cs.ox.ac.uk)}
}

\begin{document}
\maketitle

\fi

\begin{abstract}
Let $D$, called the domain, be a fixed finite set and let $\Gamma$,
called the valued constraint language, be a fixed set of functions of the
form $f:D^m\to\mathbb{Q}\cup\{\infty\}$, where different functions might have
different arity $m$. We study the \emph{valued constraint satisfaction problem}
parametrised by $\Gamma$, denoted by VCSP$(\Gamma)$. These are minimisation
problems given by $n$ variables and the objective function given by a sum of
functions from $\Gamma$, each depending on a subset of the $n$ variables.
For example, if $D=\{0,1\}$ and $\Gamma$ contains all ternary $\{0,\infty\}$-valued
functions, VCSP($\Gamma$) corresponds to 3-SAT.
More generally, if $\Gamma$ contains only $\{0,\infty\}$-valued functions,
VCSP($\Gamma$) corresponds to
CSP($\Gamma$).
If $D=\{0,1\}$ and $\Gamma$ contains all ternary $\{0,1\}$-valued
functions, VCSP($\Gamma$) corresponds to Min-3-SAT, in which the goal is to
minimise the number of unsatisfied clauses in a 3-CNF instance.
\emph{Finite-valued} constraint languages contain functions that take on
only rational values and not infinite values.
%and \emph{general-valued} constraint languages contain functions that take on rational or infinite values.
%

Our main result is a precise algebraic characterisation of valued constraint
languages whose instances can be solved \emph{exactly} by the \emph{basic linear
programming relaxation} (BLP). For a valued constraint language
$\Gamma$, BLP is a decision procedure for $\Gamma$ if and only if $\Gamma$
admits a symmetric fractional polymorphism of \emph{every} arity.
For a finite-valued constraint language $\Gamma$, BLP is a decision procedure if
and only if $\Gamma$ admits a symmetric fractional polymorphism of \emph{some}
arity, or equivalently, if $\Gamma$ admits a symmetric fractional polymorphism
of \emph{arity 2}.

Using these results, we obtain tractability of several novel classes of problems, including problems over valued constraint
languages that are: (1) submodular on \emph{arbitrary lattices}; (2)
$k$-submodular on \emph{arbitrary finite domains};
(3) weakly (and hence strongly) tree-submodular on \emph{arbitrary trees}. 
\end{abstract}

\ifARX
 \maketitle
\fi

\ifARX
\medskip\noindent
\textbf{Keywords}: 
\else
\begin{keywords}
\fi
valued constraint satisfaction, fractional polymorphisms,
submodularity, bisubmodularity, linear programming
\ifARX
\else
\end{keywords}
\fi

\ifARX
\else
% Available here:
%  http://www.ams.org/mathscinet/msc/msc2010.html?t=68Qxx&btn=Current
%
\fi

\section{Introduction}

\subsection{Constraint Satisfaction}

The constraint satisfaction problem  provides a common framework for many
theoretical and practical problems in computer science~\cite{Hell08:survey}. An
instance of the \emph{constraint satisfaction problem} (CSP) consists of a
collection of variables that must be assigned labels from a given domain subject
to specified constraints~\cite{Montanari74:constraints}. The CSP is equivalent
to the problem of evaluating conjunctive queries on
databases~\cite{Kolaitis00:jcss}, and to the homomorphism problem for relational
structures~\cite{Feder98:monotone}.

The classic 3-COLOUR problem can be seen as the following CSP: the domain
consists of three labels corresponding to the three colours; the variables
correspond to the vertices of the graph; and the constraints specify that the
variables corresponding to adjacent vertices have to be assigned different
labels. 

The CSP is NP-complete in general and thus we are interested in
restrictions which give rise to tractable classes of problems. One possibility
is to restrict the structure of the instances~\cite{Grohe07:compl,Marx13:jacm}.
Following Feder and Vardi~\cite{Feder98:monotone}, we restrict the constraint
language; that is, all constraint relations in a given instance must belong to a
fixed, finite set of relations on the domain. The most successful approach to
classifying the language-restricted CSP is the so-called algebraic
approach~\cite{Jeavons97:closure,Jeavons98:algebraic,Bulatov05:classifying},
which has led to several complexity
classifications~\cite{Bulatov06:3-elementJACM,Bulatov11:conservative,Barto09:siam,Barto11:lics}
and algorithmic characterisations~\cite{Barto14:jacm,Idziak10:siam} going beyond
the seminal work of Schaefer~\cite{Schaefer78:complexity}.

\subsection{Valued Constraint Satisfaction}

The CSP deals with only feasibility issues: Is there a solution satisfying certain
constraints? In this work we are interested in problems that capture both
feasibility and optimisation issues: What is the best solution satisfying
certain constraints? Problems of this form can be cast as valued constraint
satisfaction problems~\cite{z12:complexity,jkz14:survey}.

An instance of the \emph{valued constraint satisfaction problem} (VCSP) is given
by a collection of variables that must be assigned labels from a given domain
with the goal to \emph{minimise} the objective function that is given by the sum
of cost functions, each depending on some subset of the
variables~\cite{Cohen06:complexitysoft}. The cost functions can take on finite
rational values and positive infinity. The VCSP framework is very robust and has also
been studied under different names such as Min-Sum problems, Gibbs energy
minimisation, Markov Random Fields,  Conditional Random Fields and others in
different contexts in computer
science~\cite{Lauritzen96,Wainwright08,Crama11:book}. 

The CSP corresponds to the special case of the VCSP when the codomain of all
cost functions is $\{0,\infty\}$. Given a CSP instance, the Max-CSP consists in
determining the maximum possible number of satisfied constraints, or equivalently with
respect to exact solvability, the minimum number of unsatisfied constraints. The
Max-CSP corresponds to the case of the VCSP when the codomain of all cost
functions is $\{0,1\}$.

The VCSP is NP-hard in general and thus we are interested in the restrictions
which give rise to tractable classes of problems. As for the CSP, one can
restrict the structure of the instances~\cite{Gottlob09:icalp}. 
We will be interested in restricting the \emph{valued constraint language}; that
is, all cost functions in a given instance must belong to a fixed set of cost
functions on the domain. The ultimate goal is to understand the computational
complexity of all valued constraint languages, that is, determine which
languages give rise to classes of problems solvable in polynomial time and which
languages give rise to classes of problems that are NP-hard. Languages of the
former type are called \emph{tractable}, and languages of the latter type are
called \emph{intractable}.

Given the generality of the VCSP, it is not surprising that only few valued
constraint languages have been completely classified as tractable or
intractable. In particular, only Boolean (on a 2-element domain)
languages~\cite{Cohen06:complexitysoft,cz11:cp-mwc} and conservative (containing
all $\{0,1\}$-valued unary cost functions) languages~\cite{kz13:jacm} have been
completely classified with respect to exact solvability. 

Extending the notion of (generalised) arc consistency for the
CSP~\cite{Mackworth77:consistency,Freuder78:synthesizing} and several previously
studied notions of arc consistencies for the VCSP~\cite{Cooper04:aij-arc},
Cooper et al.\ introduced optimal soft arc consistency
(OSAC)~\cite{Cooper08:minimizing,Cooper10:osac}, which is a linear programming
relaxation of a given VCSP instance.
In fact, the VCSP problem has a natural linear programming (LP)
relaxation, proposed independently by a number of authors
\cite{Schlesinger76,Koster98,Chekuri04:sidma,Wainwright05,Carleton05,Kun12:itcs,Werner07:pami}. 
This relaxation is referred to as the \emph{basic} LP relaxation (BLP) of VCSP
as it is the first level in the Sheralli-Adams hierarchy \cite{Sherali1990},
which provides successively tighter LP relaxation of an integer LP. 
The BLP relaxation of a VCSP instance is known to be equivalent to the dual
(Lagrangian) decomposition of the instance
in which the subproblems are chosen as the individual constraints
\cite{Johnson07,Komodakis-PAMI-2011,Sontag-optbook2011}.
It is known that OSAC is at least as tight as BLP.\footnote{The difference
between BLP and OSAC is that (the dual of) OSAC has only one variable for all
constraints with the same scope (seen as a set) of variables. In BLP, different
constraints yield different BLP variables even if the scopes (seen as sets) are
the same.}

Apart from exact solvability of the CSP and its optimisation variants, the
approximability of the Max-CSP has attracted a lot of
attention~\cite{Creignouetal:siam01,Khanna01:approximability,Jonsson09:tcs}.
Under the assumption of the unique games conjecture~\cite{khot10:coco},
Raghavendra has shown that the optimal approximation ratio for the finite-valued
CSP is achieved by the basic semidefinite programming
relaxation~\cite{Raghavendra08:stoc,Raghavendra}. Recently, the classes of the
Max-CSP that are robustly approximable have been
characterised~\cite{Dalmau13:robust,Kun12:itcs,Barto12:stoc}. Specifically, Kun
et al.\ have studied the question of which classes of the Max-CSP can be
robustly approximated using BLP~\cite{Kun12:itcs}. Moreover, the power of BLP
with respect to constant-factor approximation of finite-valued CSPs has been recently
studied~\cite{Ene13:soda,Dalmau15:soda}.

More details on the complexity of the CSP can be found in~\cite{Hell08:survey}
and more details on the complexity of the VCSP can be found in the recent
survey~\cite{jkz14:survey}.

\subsection{Contributions}

We study the power of the \emph{basic linear programming relaxation} (BLP). Our
main result is a precise characterisation of valued constraint languages for
which BLP is a decision procedure. In other words, we characterise valued
constraint languages over which VCSP instances can be solved \emph{exactly} by
the BLP, i.e., when the BLP has integrality gap 1.

The characterisation is algebraic in terms of \emph{fractional
polymorphisms}~\cite{Cohen06:expressive}.
For a valued constraint language $\Gamma$ with codomain the set of rationals
with infinity, BLP is a decision procedure for $\Gamma$ if and only if $\Gamma$
admits a symmetric fractional polymorphism of \emph{every} arity.
For a valued constraint language $\Gamma$ with codomain the set of rationals
(so-called finite-valued languages),
BLP is a decision procedure if and only if any of the following equivalent
statements is satisfied: $\Gamma$ admits a symmetric fractional polymorphism of
\emph{every} arity; $\Gamma$ admits a symmetric fractional polymorphism of
\emph{some} arity; $\Gamma$ admits a symmetric fractional
polymorphism of \emph{arity 2}; $\Gamma$ admits a fractional polymorphism $\omega$ such that the
support of $\omega$ generates a symmetric operation (possibly of different arity
than the arity of $\omega$).

Our work links solving VCSP instances exactly using linear programming and the algebraic machinery for
the language-restricted VCSP introduced by Cohen et al.\ in~\cite{ccjz11:mfcs,cccjz13:sicomp}. Part of the proof is
inspired by the characterisation of the width-1
CSP~\cite{Feder98:monotone,Dalmau99:set}. The two main technical contributions
are the construction of a symmetric fractional polymorphism of a general-valued
language (Theorem~\ref{th:BLP:generate}) and the construction of symmetric
fractional polymorphisms of all arities of a finite-valued language
(Theorem~\ref{thm:finite:gen}). In order to prove these two results, we present
two techniques: a ``tree cutting'' argument, used in Section~\ref{sec:generate1}
to prove Theorem~\ref{th:BLP:generate}, and an argument based on a ``graph of
generalised operations'', used in Section~\ref{sec:finite} to prove
Theorem~\ref{thm:finite:gen}. 

Our results allow us to demonstrate that several valued constraint languages are
tractable; that is, VCSP instances over these languages can be solved exactly
using BLP. Languages not previously known to be tractable include: (1)
submodular languages on \emph{arbitrary lattices}; (2) $k$-submodular languages
on \emph{arbitrary finite domains} ($k=2$ corresponds to bisubmodularity); (3)
weakly (and hence strongly) tree-submodular languages on \emph{arbitrary trees}. 
The complexity of (subclasses of) these languages has been mentioned explicitly
as open problems
in~\cite{DeinekoJKK08,Krokhin08:max,Kolmogorov11:mfcs,Huber12:ksub}.

\subsection{Follow-up work}

Since the announcement of our results~\cite{tz12:focs}, they have already been
used to settle the complexity of the minimum 0-extension
problem~\cite{Hirai13:soda}, the complexity of the 3-element finite-valued VCSP~\cite{hkp14:sicomp},
and the complexity of the 3-element Min-Sol problems and conservative Min-Cost-Hom
problems~\cite{Uppman13:icalp}. Moreover, the last two authors have recently
shown that for finite-valued constraint languages, the condition of admitting a
symmetric fractional polymorphism of arity 2 is also necessary for tractability~\cite{tz13:stoc}.

\subsection{Combinatorial Optimisation}

Throughout the paper we assume that the objective function in our problem is
represented as a sum of functions each defined on some subset of the variables.
There is a rich tradition in combinatorial optimisation of studying problems in
which the objective function to be optimised is represented by a value-giving
\emph{oracle}. In this model, a problem is tractable if it can be solved in
polynomial time using only polynomially many queries to the oracle (where the
polynomial is in the number of variables). Any query to the oracle can
be easily simulated in linear time in the VCSP model. Consequently, a
tractability result (for a class of functions) in the value oracle model
automatically carries over to the VCSP model, while hardness results
automatically carries over in the opposite direction.

One class of functions that has received particular attention in the value
oracle model is the class of submodular functions. There are several known
algorithms for minimising a (finite-valued) submodular function using only a
polynomial number of calls to a value-giving oracle, see, for instance, Iwata's
survey~\cite{Iwata08:sfm-survey}.
Most previously discovered tractable valued constraint languages are somewhat related
to submodular functions on distributive
lattices~\cite{Cohen06:complexitysoft,Cohen08:Generalising,Jonsson11:cp,kz13:jacm}.
However, some VCSP instance with submodular functions can be solved much more
efficiently than by using these general approaches~\cite{zcj09:dam}. 

Whilst submodular functions given by an oracle can be minimised in
pseudopolynomial time on diamonds~\cite{Kuivinen11:do-diamonds}, in
polynomial
time on distributive lattices~\cite{Schrijver00:submodular,Iwata01:submodular}
and the pentagon~\cite{Krokhin08:max}, and
several constructions on lattices preserving tractability have been
identified~\cite{Krokhin08:max}, it is an interesting open open question as to
what happens on non-distributive lattices. Similarly, $k$-submodular functions
given by an oracle can be minimised in polynomial-time on domains of size
three~\cite{Fujishige06:bisubmodular}, but the complexity is unknown on domains
of larger size~\cite{Huber12:ksub}. It is known that strongly tree-submodular
functions given by an oracle can be minimised in polynomial time on binary
trees~\cite{Kolmogorov11:mfcs}, but the complexity is open on general
(non-binary) trees. Similarly, it is known that weakly tree-submodular functions
given by an oracle can be minimised in polynomial time on chains and
forks~\cite{Kolmogorov11:mfcs}, but the complexity on (even binary) trees is
open.

\section{Background}

In this section we describe the necessary background for the rest of the paper.
We start with some basic notation.
We denote $\Qnn=\{x\in\Q\:|\:x\ge 0\}$, $\Qc=\Q\cup\{\infty\}$ and $\Qnnc=\Qnn\cup\{\infty\}$.
We  define sets of real numbers $\Rnn$, $\Rc$ and $\Rnnc$ in a similar way.
Throughout the paper we assume that $0\cdot \infty=0$ and $x\cdot\infty=\infty$ for $x>0$
(we will never use such multiplication in the case when $x<0$). Note that value $\infty$ is understood as positive infinity
(and accordingly $x<\infty$ for any $x\in\R$).

\subsection{Valued CSP}

Throughout the paper, let $D$ be a fixed finite domain. We will call the
elements of $D$ \emph{labels} (for variables).
A function $f:D^m\to\Qc$ is called
an $m$-ary cost function and we say that $f$ takes \emph{values}. The argument
of $f$ is called a \emph{labelling}. For a cost function $f$, we denote $\dom f=\{x\in
D^n\:|\:f(x)<\infty\}$. % the effective domain of $f$.

A \emph{language} $\Gamma$ is a set of cost functions of possibly different
arities. A language $\Gamma$ is called \emph{finite-valued} if the codomain of
every $f\in\Gamma$ is $\Q$. If $\Gamma$ is not finite-valued
%, we assume that the codomain of cost functions $f\in\Gamma$ is $\Qc$; if needed, 
we may emphasise this fact by calling $\Gamma$ \emph{general-valued}.

\begin{definitionX}\label{def:vcsp}
An instance $\calI$ of the valued constraint
satisfaction problem (VCSP) is a function $D^V\rightarrow \Qc$ given
by
\begin{equation*}
f_\calI(x)\ =\ \sum_{t\in T} f_t(x_{v(t,1)},\ldots,x_{v(t,n_t)}). \label{eq:Cost}
\end{equation*}
It is specified by a finite set of variables $V$, finite set of terms
$T$, cost functions $f_t : D^{n_t}\rightarrow\Qc$ of arity $n_t$ and
indices $v(t, k)\in V$ for $t\in T , k =1,\ldots, n_t$. The indices of term
$t\in T$ give the \emph{scope} of the cost function $f_t$.
A solution to $\calI$ is a labelling (also an assignment) $x\in D^V$ with the
minimum total value.
The instance $\calI$ is called a $\Gamma$-instance if all terms $f_t$ belong to $\Gamma$.
\end{definitionX}

The class of optimisation problems consisting of all $\Gamma$-instances is referred to as $\VCSP(\Gamma)$.
A language $\Gamma$ is called tractable if $\VCSP(\Gamma')$ can be solved in polynomial time for \emph{each} finite $\Gamma'\subseteq\Gamma$.
It is called NP-hard if $\VCSP(\Gamma')$ is NP-hard for \emph{some} finite $\Gamma'\subseteq\Gamma$.

\subsection{The basic LP relaxation} 

Let $\mathbb M_n$ be the set of probability distributions over
labellings in $D^n$, i.e.\ $\mathbb M_n=\{\mu\ge 0\:|\:\sum_{x\in D^n}\mu(x)=1\}$.
We also denote $\Delta=\mathbb M_1$; thus, $\Delta$ is the standard ($|D|-1$)-dimensional simplex.
The corners of $\Delta$ can be identified with elements in $D$.
For a distribution $\mu\in\mathbb M_n$ and a variable $v\in\{1,\ldots,n\}$ let
 $\mu_{[v]}\in \Delta$ be the marginal probability of distribution $\mu$ for $v$:
\begin{equation*}
\mu_{[v]}(a)\ = \sum_{x\in D^n:x_v=a} \mu(x) \qquad \forall a \in D.
\end{equation*}
Given an instance $\calI$, we define the value $\BLP(\calI)$ as follows:
%Clearly, the problem of minimizing $\hat f_\calI(\alpha)$ can be formulated as an LP of polynomial size (assuming that $\Gamma$ is finite). 
\begin{eqnarray}
&& \hspace{-85pt} \BLP(\calI)\ =\ \min\ \sum_{t\in T}\sum_{x\in \dom f_t}\mu_t(x)f_t(x)  \label{eq:BLP} \\
 \mbox{s.t.~~} (\mu_t)_{[k]}&=&\alpha_{v(t,k)} \hspace{20pt} \forall t\in T,k\in\{1,\ldots,n_t\} \nonumber \\
\mu_t&\in&\mathbb M_{n_t}                      \hspace{29pt} \forall t\in T \nonumber \\
\mu_t(x)&=&0                                   \hspace{43pt} \forall t\in T,x\notin \dom f_t \nonumber \\
\alpha_v&\in&\Delta                            \hspace{40pt} \forall v\in V \nonumber 
\end{eqnarray}
If there are no feasible solutions then $\BLP(\calI)=\infty$.
All constraints in this system are linear, therefore this is a linear program. 
We call it the \emph{basic LP relaxation} of $\calI$ (BLP).
We say that BLP \emph{solves} $\calI$ if 
$\BLP(\calI)=\min_{x\in D^n}f_\calI(x)$.
We say that BLP solves a language $\Gamma$ if it solves all instances
$\calI\in\VCSP(\Gamma)$.

\subsection{Fractional polymorphisms}

We denote by $\calO^{(m)}$ the set of $m$-ary operations $g:D^m\rightarrow D$.
An operation $g:D^2\to D$ of arity 2 is called \emph{binary}.
An $m$-ary projection on the $i$th coordinate is the operation
$e^{(m)}_i:D^m\rightarrow D$ defined by $e^{(m)}_i(x_1,\ldots,x_m)=x_i$.
Let $S_m$ be the symmetric group on $\{1,\dots,m\}$. An operation
$g\in\calO^{(m)}$ is called \emph{symmetric} if it is invariant with respect to
any permutation of its arguments:
$g(x_1,\ldots,x_m)=g(x_{\pi(1)},\ldots,x_{\pi(m)})$ for any permutation $\pi \in
S_m$ and any $(x_1,\ldots,x_m)\in D^m$. 
The set of symmetric
operations in $\calO^{(m)}$ will be denoted by $\calO^{(m)}_{\tt sym}$.

A \emph{fractional operation} of arity $m$ is a vector $\omega:\fpol$
satisfying $\| \omega \|_1 = 1$, where $\|\omega\|_1 = \sum_{g\in\calO^{(m)}} \omega(g)$.
We let $\supp(\omega)$ denote the support of $\omega$, defined by
$\supp(\omega)=\{g\in\calO^{(m)}\:|\:\omega(g)>0\}$.

It will often be convenient to write a fractional operation $\omega:\fpol$ as a sum
$\omega=\sum_{g\in\calO^{(m)}} \omega(g) \cdot \chi_g$, where
$\chi_g:\fpol$ denotes the vector
that assigns weight $1$ to the operation $g$ and $0$ to all other operations.

A fractional operation $\omega$ is called symmetric %(cyclic, respectively)
if all operations in $\supp(\omega)$ are symmetric. % (cyclic, respectively).

The \emph{superposition}, $h[g_1,\dots,g_n]$,  of an $n$-ary operation $h$ with
$n$ $m$-ary operations $g_1, \dots, g_n$ is the $m$-ary operation defined by 
\[ 
h[g_1,\dots,g_n](x_1,\dots,x_m) =
h(g_1(x_1,\dots,x_m),\dots,g_n(x_1,\dots,x_m)).
\]
This can also be seen as a composition $h\circ(g_1,\dots,g_n) : D^m \to D$ of
the operation 
$h:D^n\rightarrow D$ and the mapping $(g_1,\dots,g_n):D^m\rightarrow D^n$.

The \emph{superposition}, $\omega[g_1,\dots,g_n]$, of an $n$-ary fractional
  operation $\omega$ with $n$ $m$-ary operations $g_1,\dots,g_n$ is the $m$-ary
 fractional operation defined as follows: %%$\omega'$, where
  \[
  \omega[g_1,\dots,g_n](h) = \sum_{\{h' \:\mid\: h = h'[g_1,\dots,g_n]\}} \omega(h').
  \]

The following example illustrates these definitions.

\begin{exampleX}\examplefont\label{ex:sub1}
Let $D=\{0,1,\ldots,d\}$ and let $\min,\max:D^2\to D$ be the two binary
operations on $D$ that return the smaller (larger) of its two arguments
respectively with respect to the natural order of integers.

The ternary operation $\min^{(3)}$ returning the smallest of its three arguments can be
obtained by the following superposition:
$\min^{(3)}(x,y,z)=\min[e^{(3)}_1,\min[e^{(3)}_2,e^{(3)}_3]]$.

Let $\omega$ be the fractional operation that assigns weight $\frac{1}{2}$ to
$\min$ and weight $\frac{1}{2}$ to $\max$. Clearly,
$\supp(\omega)=\{\min,\max\}$. Since both $\min$ and $\max$ are
symmetric operations, we have that $\omega$ is a symmetric fractional operation.

Let $\min^{(4)}_{12}$ and $\min^{(4)}_{34}$ be the two 4-ary operations that
return the smaller of its first (last) two arguments respectively. Then the
superposition of $\omega$ with $\min^{(4)}_{12}$ and $\min^{(4)}_{34}$ is the
4-ary fractional operation $\omega'=\omega[\min^{(4)}_{12},\min^{(4)}_{34}]$ that assigns
weight $\frac{1}{2}$ to the operation $\min^{(4)}$, which returns the smallest
of its four arguments, and weight $\frac{1}{2}$ to the operation
$\max[\min^{(4)}_{12},\min^{(4)}_{34}]$.
\end{exampleX}

\begin{definitionX}
For an $n$-ary cost function $f:D^n\to\Qc$ and $x^1,\ldots,x^m\in D^n$ for some $m\geq 1$, we
define the average value of $f$ applied to the labellings $x^1,\ldots,x^m$
by 
\begin{equation*}
f^m(x^1,\ldots,x^m)\ =\ \frac{1}{m}(f(x^1)+\ldots+f(x^m)).
\end{equation*}
\end{definitionX}

\begin{definitionX}\label{def:FP}
A fractional operation $\omega:\fpol$ is called a \emph{fractional polymorphism}
of the language $\Gamma$, and we say that $\Gamma$ \emph{admits} $\omega$, 
if for every cost function $f \in \Gamma$,
\begin{equation}
\sum_{g\in\calO^{(m)}} \omega(g)f(g(x^1,\ldots,x^m))\ \le\
f^m(x^1,\ldots,x^m)\qquad\forall x^1,\ldots,x^m\in \dom f.
\label{eq:FPdefinition}
\end{equation}
\end{definitionX}

We note that~\eqref{eq:FPdefinition} implies that if $g\in \supp(\omega)$
and $x^1,\ldots,x^m\in \dom f$ then $g(x^1,\ldots,x^m)\in \dom f$.

Definition~\ref{def:FP} is illustrated in
Figure~\ref{fig:FP}, which should be read from left to right. Let $f$ be an
$n$-ary cost function and $\omega:\fpol$ an $m$-ary fractional operation. 
Moreover, let $k=|\calO^{(m)}|$.
Starting with $m$ $n$-tuples $x^1,\ldots,x^m\in\dom f$, which we view as row
vectors in Figure~\ref{fig:FP}, we first apply all $m$-ary operations $g_1,\ldots,g_k$
 to these tuples componentwise, thus obtaining
the $m$-tuples $y^1,\ldots,y^k$. Inequality~\ref{eq:FPdefinition}
amounts to comparing the average of the values of $f$ applied to the tuples
$x^1,\ldots,x^m$ with
the weighted sum of the values of $f$ applied to the tuples
$y^1,\ldots,y^k$, where the weight of the $i$th tuple $y^i$ (obtained from
$g_i$) is the
weight assigned to $g_i$ by $\omega$.
\begin{figure}[hbtp]
\small
\[
\begin{array}{c}
\begin{array}{c}
{x^1}\\
{x^2}\\
\vdots\\
{x^m}\\
\end{array}
\\
\begin{array}{c}
\
\end{array}
\\
\begin{array}{c}
{y^1}=g_1({x^1},\ldots,{x^m})\\
{y^2}=g_2({x^1},\ldots,{x^m})\\
\vdots\\
{y^k}=g_k({x^1},\ldots,{x^m})\\
%t_n'=f_n(t_1,\ldots,t_k)\\
\end{array}
\end{array}
\begin{array}{c}
\begin{array}{cccccc}
=(x_1^1 & x_2^1 & \ldots & x_n^1) \\
=(x_1^2 & x_2^2 & \ldots & x_n^2) \\
& & \vdots & \\
=(x_1^m & x_2^m & \ldots & x_n^m) \\
\end{array}
\\
%\hline
\begin{array}{c}
\
\end{array}
\\
\begin{array}{cccccc}
=(y_1^1 & y_2^1 & \ldots & y_n^1) \\
=(y_1^2 & y_2^2 & \ldots & y_n^2) \\
 & & \vdots & \\
=(y_1^k & y_2^k & \ldots & y_n^k) \\
\end{array}
\\
\end{array}
\begin{array}{l}
\stackrel{f}{\longrightarrow}
\left.
\begin{array}{c}
f({x^1})\\
f({x^2})\\
\vdots\\
f({x^m})\\
\end{array}
\right\}\mbox{\normalsize{$\displaystyle \frac{1}{m}\sum_{i=1}^{m}f({x^i})$}}
\\
\begin{array}{c}
\qquad\qquad\qquad\qquad \mbox{\rotatebox{270}{$\geq$}}
\end{array}
\\
\stackrel{f}{\longrightarrow}
\left.
\begin{array}{c}
f({y^1})\\
f({y^2})\\
\vdots\\
f({y^k})\\
\end{array}
\right\}\mbox{\normalsize{$\displaystyle \sum_{i=1}^{k}\omega(g_i)f(y^i)$}}
\\
\end{array}
\]
\caption{Definition of a fractional polymorphism.}\label{fig:FP}
\end{figure}

%%
%
%%We view labellings in $D^n$ as column vectors in Figure~\ref{fig:fpol}, where $n$ is the arity of
%%$f\in\Gamma$ in~(\ref{eq:FPdefinition}); given $m$ such columns,
%%an operation $g:D^m\rightarrow D$ produces a new column as shown on the right in
%%Figure~\ref{fig:fpol}.
%%%
%%%\begin{wrapfigure}{r}{0.32\textwidth}
%%\begin{figure}[htb]
%%%\vspace{-25pt}
%%\begin{center}
%%\begin{tabular}{ccc@{\hspace{0pt}}}
%%$x^1$ &  $\dots$ & $x^m$ \\
%%\hline
%%$x^1_1$  & $\dots$ & $x^m_1$ \\
%%$\vdots$  & $\ddots$ & $\vdots$ \\
%%$x^1_n$ & $\dots$ & $x^m_n$ 
%%\end{tabular}
%%\qquad 
%%\begin{tabular}{c}
%%$g(x^1,\dots,x^m)$ \\
%%\hline
%%$g(x^1_1,\dots,x^m_1)$ \\
%%$\vdots$ \\
%%$g(x^1_n,\dots,x^m_n)$ 
%%\end{tabular}
%%%\vspace{-5pt}
%%%\end{wrapfigure} 
%%\end{center}
%%\caption{Definition of a fractional polymorphism.}\label{fig:fpol}
%%\end{figure}
%%

%\begin{remarkX}\examplefont
%We remark that  we define fractional polymorphisms in two
%different contexts: (i) fractional operations that assign total weight 1, and
%(ii) (property of a language $\Gamma$) fractional operations that assign total
%weight 1 and satisfy~(\ref{eq:FPdefinition}) for all $f\in\Gamma$. While only
%(ii) is the standard meaning in the literature, we find (i) useful in
%presenting our results.
%\end{remarkX}

\begin{exampleX}\examplefont
A simple example of an $m$-ary fractional operation is the 
vector $\omega:\fpol$ defined by $\omega(e^{(m)}_i)=1/m$ for all $1\leq i\leq m$ and $\omega(h)=0$ for any
$m$-ary operation $h$ that is not a projection.
%Since there are precisely $m$
%$m$-ary projections, $\|\omega\|_1=1$ and hence $\omega$ is a fractional
%polymorphism.
It follows from
Definition~\ref{def:FP} that $\omega$ is a fractional polymorphism of every cost function $f$
and in fact~(\ref{eq:FPdefinition}) holds with equality in this case.
\end{exampleX}

\begin{exampleX}\examplefont
Recall
from Example~\ref{ex:sub1}
the fractional operation $\omega$ defined on $D$ that assigns weight $\frac{1}{2}$ to $\min$
and weight $\frac{1}{2}$ to $\max$. 
%By definition, $\|\omega\|_1 = 1$ and thus
%$\omega$ is a fractional polymorphism. 
In this special case, 
(\ref{eq:FPdefinition}) simplifies to
\begin{equation*}
f(\min(x^1,x^2)) + f(\max(x^1,x^2))\ \le\ f(x^1) + f(x^2) \qquad\forall
x^1,x^2\in \dom f.
\end{equation*}

For $D=\{0,1\}$, a cost
function $f$ that satisfies (\ref{eq:FPdefinition}) with this $\omega$ is called
\emph{submodular}~\cite{Schrijver00:submodular}. 
\end{exampleX}

\begin{remarkX}\examplefont\label{rem:prob}
One can equivalently view fractional polymorphisms in a probabilistic setting.
A fractional operation $\omega:\fpol$ is a fractional polymorphism of $\Gamma$
if $\omega$ is a probability distribution over $\calO^{(m)}$,
and every cost function $f\in\Gamma$
satisfies,
\begin{equation*}
\displaystyle\mathbb{E}_{g\sim\omega} f(g(x^1,\ldots,x^m)) \ \le\
f^m(x^1,\ldots,x^m)\qquad\forall x^1,\ldots,x^m\in \dom f.
\label{eq:FPdefinition2} 
\end{equation*}
%
%We will use both views of fractional polymorphisms in the rest of the
%paper.\merk{Standa: keep in?}
\end{remarkX}

%%A \emph{fractional operation} of arity $m$ is a probability distribution
%%$\omega$ over $\calO^{(m)}$, i.e., $\omega : \calO^{(m)}
%%\to \Rnn$ so that $\|\omega\|_1 := \sum_{g} \omega(g) = 1$.
%%We let $\supp(\omega)$ denote the support of $\omega$:
%%$\supp(\omega)=\{g\in\calO^{(m)}\:|\:\omega(g)>0\}$.

%%In general $\omega'$ is not a fractional polymorphism, 
%%even when $\omega$ is,
%%but $\|\omega'\|_1 = \|\omega\|_1$ and $\omega'$ satisfies the following inequality:
%%\begin{align}\label{eq:superpos}
%%\sum_{h' \in \calO^{(m)}} \omega'(h') f(h'(x^1,\dots,x^m))
%%&= \sum_{h \in \calO^{(n)}} \omega(h) f(h(g_1,\dots,g_n)(x^1,\dots,x^m)) \notag\\
%%&\leq \frac{1}{n} \sum_{i=1}^n f(g_i(x^1,\dots,x^m)),
%%\end{align}
%%for every $f\in\Gamma$ and $x^1,\dots, x^m\in \dom f$.

\section{Results}

In this section we will state our main results together with some algorithmic
consequences. The rest of the paper will be devoted to the proofs of the
results.

\subsection{The power of BLP}
Let $\Gamma$ be a language such that the set $\Gamma$ is countable. 
First, we give a precise characterisation of the power of BLP for
(general-valued) languages.

%\newcounter{mycounterBLPgeneral}
%\setcounter{mycounterBLPgeneral}{\value{theorem}}

\begin{theoremX}\label{th:BLP:general}
BLP solves $\Gamma$ if and only if $\Gamma$ admits a symmetric fractional
polymorphism of \emph{every} arity $m\ge 2$.
\end{theoremX}

Second, we give a sufficient condition for the existence of symmetric fractional
polymorphisms. But first we need a standard definition from universal algebra.

A set $\calC$ of operations is called a \emph{clone} if it contains all
projections and is closed under superposition; that is, $\calC$ contains
$e^{(m)}_i$ for all $m\geq 1$ and $1\leq i\leq m$, and if
$h,g_1,\ldots,g_n\in\calC$ then $h[g_1,\ldots,g_n]\in\cal C$, where $h$ is an
$n$-ary operation and $g_1,\ldots,g_n$ are operations of the same
arity.
A set $\calO$ of operations is said to \emph{generate} $g$ if $g$ belongs to the
smallest clone containing $\calO$; in other words, $g$ can be obtained by
superpositions of operations from $\calO$ and projections. 

\begin{exampleX}\label{ex:max}\examplefont
If $\calO$ contains the binary maximum operation $\max:D^2\to D$ that returns
the larger of its two arguments (with respect to some total order on $D$), then
$\calO$ can generate the $m$-ary operation $\max^{(m)}$ that returns the largest of
its $m$ arguments by
$\max^{(m)}(x_1,\ldots,x_m)=\max(x_1,\max(x_2,\ldots,\max(x_{m-1},x_m)\ldots)).$
\end{exampleX}

\begin{theoremX}\label{th:BLP:generate}
Suppose that, for every $n \geq 2$, 
$\Gamma$ admits a fractional polymorphism $\omega_n$
such that $\supp(\omega_n)$ generates a symmetric $n$-ary operation. 
Then, $\Gamma$ admits a symmetric fractional polymorphism of every arity $m \geq 2$.
\end{theoremX}

The following is an immediate consequence of Theorems~\ref{th:BLP:general}
and~\ref{th:BLP:generate}.

\begin{corollaryX}
BLP solves $\Gamma$ if and only if for every $n\ge 2$, $\Gamma$ admits a
fractional polymorphism $\omega_n$ such that $\supp(\omega_n)$ generates a
symmetric $n$-ary operation.
\end{corollaryX}
%%\begin{proof}
%%The ``if'' part follows from Theorem~\ref{th:BLP:generate}, whereas the ``only
%%if'' part follows from Theorem~\ref{th:BLP:general}.
%%\end{proof}

Finally, we give, in Theorem~\ref{thm:tfae}, a more refined characterisation of the power of BLP for
finite-valued languages. %%, which is simpler than the one for general-valued languages given in Theorem~\ref{th:BLP:general}.
It is based on the following result.

\begin{theoremX}\label{thm:finite:gen}
Suppose that a finite-valued language $\Gamma$ admits
a symmetric fractional polymorphism of arity $m-1\ge 2$. Then
$\Gamma$ admits a symmetric fractional polymorphism of arity $m$.
\end{theoremX}

\begin{theoremX} \label{thm:tfae}
Suppose that $\Gamma$ is finite-valued. The following are equivalent:
\begin{enumerate}
\item
BLP solves $\Gamma$;
\item
$\Gamma$ admits a symmetric fractional polymorphism of \emph{every} arity $m \geq 2$;
\item
$\Gamma$ admits a symmetric fractional polymorphism of \emph{some} arity $m \geq 2$;
\item
$\Gamma$ admits a symmetric fractional polymorphism of \emph{arity 2};
\item
For every $n \geq 2$, $\Gamma$ admits a fractional polymorphism $\omega_n$ such that
$\supp(\omega_n)$ generates a symmetric $n$-ary operation.
\end{enumerate}
\end{theoremX}

\begin{proof}
The equivalence between statements (1) and (2) is a special case of Theorem~\ref{th:BLP:general}.
The implications $(2) \implies (4) \implies (3)$ and $(2) \implies (5)$ are trivial.
The implication $(4) \implies (2)$ follows from Theorem~\ref{thm:finite:gen}.
Assume $(3)$. By Theorem~\ref{thm:finite:gen}, we may assume that $m$ is even. 
Let $e^{(2)}_1$ and $e^{(2)}_2$ be the two binary projections on the domain of $\Gamma$.
Then $\omega[e^{(2)}_1,e^{(2)}_2,\dots,e^{(2)}_1,e^{(2)}_2]$ is a binary symmetric fractional polymorphism of $\Gamma$, so (4) follows.
Finally, the implication $(5) \implies (3)$ follows from Theorem~\ref{th:BLP:generate}.
\end{proof}

Note that the finite-valuedness assumption in Theorems~\ref{thm:finite:gen} and~\ref{thm:tfae}
is essential: for general-valued languages these theorems do not hold as the
following example demonstrates.

\begin{exampleX}\label{example:STPcycle}\examplefont
Let $D=\{a,b,c\}$ and consider the binary operation $g:D^2\to D$ defined by
$g(x,x)=x$ for $x\in D$,
$g(a,b)=g(b,a)=b$, $g(b,c)=g(c,b)=c$, and $g(a,c)=g(c,a)=a$ ($g$ corresponds to
the oriented cycle $a\to b\to c\to a$). 
Note that $g$ is symmetric. Moreover, $g$ is also conservative, that is,
$g(x,y)\in\{x,y\}$ for all $x,y\in D$. Any operation that is symmetric and
conservative is called a tournament operation~\cite{Cohen08:Generalising}.
Consider the fractional operation $\omega$ defined by $\omega(g)=1$. It is
known that any general-valued constraint language 
admitting $\omega$ is tractable~\cite{Cohen08:Generalising} ($\omega$ is called
a tournament pair in~\cite{Cohen08:Generalising}).
%
%Let $D=\{a,b,c\}$ and consider the binary operations $g,g':D^2\to D$ defined by
%$$
%(g(x,y),g'(x,y))=
%\begin{cases}
%(x,x) & \mbox{if }x=y \\
%(a,b) & \mbox{if }\{x,y\}=\{a,b\} \\
%(b,c) & \mbox{if }\{x,y\}=\{b,c\} \\
%(c,a) & \mbox{if }\{x,y\}=\{c,a\} 
%\end{cases}\qquad \forall x,y\in D
%$$
%The pair $\mm{g}{g'}$ of this form is called a \emph{symmetric tournament pair}~\cite{Cohen08:Generalising}
%(it corresponds to the oriented cycle $a\to b\to c\to a$).
%Consider the fractional operation $\omega$ defined by $\omega(g_1)=\omega(g_2)=\frac{1}{2}$. It is
%known that any general-valued constraint language 
%admitting $\omega$ is tractable~\cite{Cohen08:Generalising}.

Let $f:D^2\to\Qnnc$ be the following binary cost function: $f(x,y)=0$ if
$(x,y)\in\{(a,b),(b,c),(c,a)\}$ and $f(x,y)=\infty$ otherwise. Let
$\Gamma=\{f\}$. It can be verified that $\Gamma$ admits $\omega$ as a fractional
polymorphism and thus is tractable.

We now show, however, that $\Gamma$ does not admit any ternary symmetric
fractional polymorphism. Let $h:D^3\to D$ be an arbitrary ternary symmetric
operation. Since $\dom f=\{(a,b),(b,c),(c,a)\}$, we have that $h$ applied to the
tuples $(a,b)$, $(b,c)$, and $(c,a)$ componentwise gives a tuple $(x,x)$ for
some $x\in D$ but $(x,x)\not\in\dom f$ for any $x\in D$. Thus no ternary
fractional polymorphism of $\Gamma$ can have a symmetric ternary operation in
its support. By Theorem~\ref{th:BLP:general}, BLP does \emph{not} solve
$\Gamma$.
\end{exampleX}

\subsection{Examples of languages solved by BLP} %% tractable languages}
\label{subsec:examples}

We now give examples of languages that are solved by BLP.
In some cases, the tractability of these languages was known before,
while in others, we present here the first proof of their tractability.

A binary operation $g:D^2\to D$ is \emph{idempotent} if $g(x,x)=x$ for all $x\in D$,
\emph{commutative} if $g(x,y)=g(y,x)$ for all $x,y\in D$, and \emph{associative} if
$g(x,g(y,z))=g(g(x,y),z))$ for all $x,y,z\in D$. A binary operation $g:D^2\to D$
is a \emph{semilattice operation} if $g$ is idempotent, commutative, and
associative.
The $\max$-operation of Example~\ref{ex:max} is an example of a semilattice operation.
In the same way that $\max$ generates $\max^{(m)}$, $m \geq 2$,
every semilattice operation $g:D^2\to D$
generates symmetric operations of all arities.
In particular, a symmetric
operation $g^{(m)}:D^m\to D$ can be obtained from $g$ by
\[
g^{(m)}(x_1,\ldots,x_m) = g(x_1,g(x_2,\ldots,g(x_{m-1},x_m)\ldots)).
\]
Consequently, we obtain the following result.

\begin{corollaryX}[of Theorem~\ref{th:BLP:general} and Theorem~\ref{th:BLP:generate}]\label{cor:tractability}
If $\Gamma$ admits a fractional polymorphism with
a semilattice operation in its support, then BLP solves $\Gamma$.
\end{corollaryX}

Most previously identified tractable languages have been defined via \emph{binary
multimorphisms}, which are a special case of binary fractional
polymorphisms~\cite{Cohen06:complexitysoft}.
A binary multimorphism $\mm{g_1}{g_2}$ of a language $\Gamma$ is a binary
fractional polymorphism $\omega$ of $\Gamma$ such that
$\omega(g_1)=\omega(g_2)=1/2$, where $g_1,g_2:D^2\rightarrow D$. For a binary
multimorphism $\mm{g_1}{g_2}$, the fractional polymorphism
inequality~(\ref{eq:FPdefinition}) simplifies to, %% for every cost function $f \in \Gamma$,% of arity $n$,
%
%\begin{equation}
\[
f(g_1(x^1,x^2)) + f(g_2(x^1,x^2))\ \leq\ f(x^1) + f(x^2)\qquad \forall f \in \Gamma, x^1,x^2 \in \dom f.
\]
%\label{eq:MMdefinition}
%\end{equation}

With the exception of skew bisubmodularity, the languages discussed below are
all defined by binary multimorphisms.

\medskip

%\begin{description}

%\item[Submodularity on a lattice] \hfill\\
\subsection*{Submodularity on a lattice}
Let $(D; \meet, \join)$ be an \emph{arbitrary} lattice on $D$, where $\meet$ and
$\join$ are the meet and join operations, respectively. Let $\Gamma$ be a
language admitting the multimorphism $\mm{\meet}{\join}$; such languages are
called \emph{submodular} on the lattice $(D;\meet,\join)$. 
The operations $\meet$ and $\join$ of any lattice are semilattice operations,
hence
Corollary~\ref{cor:tractability} shows that BLP solves $\Gamma$. The
tractability of submodular languages was previously known only for distributive
lattices~\cite{Schrijver00:submodular,Iwata01:submodular}.
Moreover, several
tractability-preserving operations on lattices have been identified 
in~\cite{Krokhin08:max}.
Finally, it is known that VCSP instances over submodular languages on diamonds
can be minimised in pseudopolynomial time~\cite{Kuivinen11:do-diamonds}.

%\item[Symmetric tournament pair]\hfill\\ 
\subsection*{Symmetric tournament pair}
A binary operation $g:D^2\to D$ is
conservative if $g(x,y)\in\{x,y\}$ for all $x,y\in D$. A binary operation
$g:D^2\to D$ is a tournament operation if $g$ is commutative and conservative.
The dual of a tournament operation $g$ is the unique tournament operation $g'$
satisfying $g(x,y)\neq g'(x,y)$ for all $x\neq y$. The multimorphism
$\mm{g_1}{g_2}$ is a \emph{symmetric tournament pair} (STP) if both $g_1$ and
$g_2$ are tournament operations and $g_2$ is the dual of
$g_1$~\cite{Cohen08:Generalising}.
If $\Gamma$ is a finite-valued language with an STP multimorphism
$\mm{g_1}{g_2}$ then $\Gamma$ also admits a submodularity multimorphism
discussed above. This result is implicitly contained
in~\cite{Cohen08:Generalising} and a full proof is given in
Appendix~\ref{sec:STP}. Consequently, BLP solves $\Gamma$ by
Corollary~\ref{cor:tractability}. This also follows from Theorem~\ref{thm:tfae}.

%\item[$k$-Submodularity]\hfill\\
\subsection*{$k$-Submodularity}
Let $D=\{0,1,\ldots,k\}$ and let $\Gamma$ be a language defined on $D$ that
admits the multimorphism $\mm{\min_0}{\max_0}$~\cite{Cohen06:complexitysoft},
where $\min_0(x,x)=x$ for all $x\in D$ and $\min_0(x,y)=0$ for all $x,y\in D,
x\neq y$; $\max_0(x,y)=0$ if $0\neq x\neq y\neq 0$ and $\max_0(x,y)=\max(x,y)$
otherwise, where $\max$ returns the larger of its two arguments with respect to
the normal order of integers; such languages are known as
\emph{$k$-submodular}~\cite{Huber12:ksub}.
Since $\min_0$ is a semilattice operation, BLP solves $\Gamma$ by
Corollary~\ref{cor:tractability}. 
The tractability of $k$-submodular languages was previously open for
$k>2$~\cite{Huber12:ksub}.

Applications of $k$-submodular functions can be found in~\cite{Gridchyn:ICCV13,Wahlstroem:SODA14}.

%\item[Bisubmodularity]\hfill\\
\subsection*{Bisubmodularity}
The special case of $k$-submodularity for $k=2$ is known as
\emph{bisubmodularity}. The tractability of (finite-valued) bisubmodular
languages was previously known only using a general algorithm for minimising
bisubmodular set
functions~\cite{Fujishige06:bisubmodular,McCormick10:bisubmodular}. 

%\item[Skew bisubmodularity]\hfill\\
\subsection*{Skew bisubmodularity}
Let $D=\{0,1,2\}$ with the partial order satisfying $0<1$ and $0<2$. Recall the definition of
the operations $\min_0$ and $\max_0$ from the description of $k$-submodularity
above. We define
$\max_1(x,y)=1$ if $0\neq x\neq y\neq 0$ and $\max_1(x,y)=\max(x,y)$ otherwise,
where $\max$ returns the larger of its two arguments with respect to the normal
order of integers.
A language $\Gamma$ defined on $D$ is called \emph{$\alpha$-bisubmodular}, for some
real $0<\alpha\leq 1$, if $\Gamma$ admits a fractional polymorphism $\omega$
defined by $\omega(\min_0)=1/2$, $\omega(\max_0)=\alpha/2$, and
$\omega(\max_1)=(1-\alpha)/2$. 
(Note that $1$-bisubmodular languages are bisubmodular languages 
discussed above.)
A language that is $\alpha$-bisubmodular for some $\alpha$ is called
\emph{skew bisubmodular}.
Since $\min_0$ is a semilattice operation, BLP solves $\Gamma$ by Corollary~\ref{cor:tractability}.
The tractability of skew bisubmodular languages was first observed
in~\cite{hkp14:sicomp} using an extended abstract of this paper~\cite{tz12:focs}.
%
%Interestingly, it can be proved that skew bisubmodular languages cannot be defined by
%multimorphisms~\cite{hkp14:sicomp}.

%\item[Strong tree-submodularity]\hfill\\
\subsection*{Strong tree-submodularity}
Assume that the labels in the domain $D$ are arranged into a tree $T$. The tree
induces a partial order: $a\preceq b$ if $a$ is an ancestor of $b$, that is, if
$a$ lies on the unique path from $b$ to the root of $T$. 
Given
$a,b\in T$, let $P_{ab}$ denote the unique path in $T$ between $a$ and $b$ of
length (=number of edges) $d(a,b)$, and let $P_{ab}[i]$ denote the $i$-th vertex
on $P_{ab}$, where $0\leq i\leq d(a,b)$ and $P_{ab}[0]=a$. Let $\tuple{g_1,g_2}$
be two binary commutative operations defined as follows: given $a$ and $b$,
let $a_1=P_{ab}[\lfloor d/2\rfloor]$ and $a_2=P_{ab}[\lceil d/2\rceil]$. If
$a_2\preceq a_1$ then swap $a_1$ and $a_2$ so that $a_1\preceq a_2$. Finally,
$g_1(a,b)=g_1(b,a)=a_1$ and $g_2(a,b)=g_2(b,a)=a_2$.
Let $\Gamma$ be a language admitting the multimorphism $\tuple{g_1,g_2}$; such
languages are called \emph{strongly tree-submodular}. Since $g_1$ is a
semilattice operation, BLP solves $\Gamma$ by
Corollary~\ref{cor:tractability}. The tractability of finite-valued strongly
tree-submodular languages on binary trees has been shown
in~\cite{Kolmogorov11:mfcs} but the tractability of strongly tree-submodular
languages on non-binary trees was left open.

%\item[Weak tree-submodularity]\hfill\\
\subsection*{Weak tree-submodularity}
Assume that the labels in the domain $D$ are arranged into a tree $T$.
For $a,b\in T$, let $g_1(a,b)$ be defined as the highest common ancestor of $a$
and $b$ in $T$; that is, the unique node on the path $P_{ab}$ that is ancestor
of both $a$ and $b$.  We define $g_2(a,b)$ as the unique node on the path
$P_{ab}$ such that the distance between $a$ and $g_2(a,b)$ is the same as the
distance between $b$ and $g_1(a,b)$.
Let $\Gamma$ be a language admitting the multimorphism $\tuple{g_1,g_2}$; such
languages are called \emph{weakly tree-submodular}.
Since $g_1$ is a semilattice operation, BLP solves $\Gamma$ by
Corollary~\ref{cor:tractability}. The tractability of finite-valued weakly
tree-submodular languages on chains\footnote{A chain is a binary tree in which
all nodes except leaves have exactly one child.} and forks\footnote{A fork is a
binary tree in which all nodes except leaves and one special node have exactly
one child. The special node has exactly two children.} has been shown
in~\cite{Kolmogorov11:mfcs} and left open for all other trees. Weak
tree-submodularity generalises the above-discussed concept of strong
tree-submodularity in the sense that any language that is strongly
tree-submodular is also weakly tree-submodular~\cite{Kolmogorov11:mfcs}. Weak
tree-submodularity also generalises the above-discussed concept of
$k$-submodularity, which corresponds to the special case with a tree on $k+1$
vertices consisting of a root node with $k$ children.

%\item[1-Defect chain]\hfill\\
\subsection*{$1$-Defect chain}
In our final example, Corollary~\ref{cor:tractability} does not suffice
to prove that BLP solves the specific languages.
Instead, we refer directly to Theorems~\ref{th:BLP:general}
and~\ref{th:BLP:generate} (and also Theorem~\ref{thm:tfae} in the special case
of finite-valued languages).
Let $b$ and $c$ be two distinct elements of $D$ and let $(D;<)$ be a partial order
which relates all pairs of elements except for $b$ and $c$. A pair 
$\tuple{g_1,g_2}$, where $g_1,g_2:D^2\rightarrow D$ are two binary operations, is a
\emph{1-defect chain} multimorphism if $g_1$ and $g_2$ are both commutative and satisfy the following
conditions:
\begin{itemize}
\item If $\{x,y\}\neq\{b,c\}$, then $g_1(x,y)=x\meet y$ and $g_2(x,y)=x\join y$.
\item If $\{x,y\}=\{b,c\}$, then $\{g_1(x,y),g_2(x,y)\}\cap\{x,y\}=\emptyset$, and
$g_1(x,y)<g_2(x,y)$.
\end{itemize}

The tractability of finite-valued languages admitting a 1-defect chain multimorphism has
been shown in~\cite{Jonsson11:cp}. 
By Theorem~\ref{thm:tfae}, the BLP solves any such language.
We now show a more general result: BLP also solves general-valued languages admitting a
1-defect chain multimorphism.

We consider the case when $g_1(b,c)<b,c$ and set $g=g_1$.
(An analogous argument works in the case when $g_2(b,c)>b,c$.)
Using $g$, we construct a symmetric $m$-ary operation $h^{(m)}(x_1,\dots,x_m)$ for each
$m$. Consequently, BLP solves $\Gamma$ by Theorem~\ref{th:BLP:general} and
Theorem~\ref{th:BLP:generate}.

Let $h_1,\ldots,h_M$ be the $M=\binom{m}{2}$ terms $g(x_i,x_j)$.
Let
\[
h^{(m)}=g(h_1,g(h_2,\dots,g(h_{M-1},h_M)\dots)).
\]
There are three possible cases:

\begin{itemize}

\item 

$\{b,c\} \not\subseteq {x_1, \dots,x_m}$. Then $g$ acts as $\meet$, which is a
semilattice operation, hence so does $h^{(m)}$.

\item

$\{b,c\} \subseteq \{x_1, \dots,x_m\}$ and $g(b,c)\leq x_1,\dots,x_m$. Then
$h_i=g(b,c)$ for some $1\leq i\leq M$, and $g(h_i,h_j)=g(b,c)$ for all $1\leq
j\leq M$, so $h^{(m)}(x_1,\dots,x_m)=g(b,c)$.

\item

$\{b,c\} \subseteq \{x_1,\dots,x_m\}$ and $x_p \leq g(b,c)$, for some $1\leq p \leq m$. 
By choice of $g$, $x_p\not\in\{b,c\}$ and we can additionally choose $p$ so that
$x_p\leq x_1,\ldots,x_m$. Then $g(x_p,x_q)=x_p$ for all $1\leq q\leq m$ so
$h_i=x_p$ for some $1\leq i\leq M$ and $g(h_i,h_j)=x_p$ for all $1\leq j\leq M$,
so $h^{(m)}(x_1,\dots,x_m)=x_p$.

\end{itemize}

%%%%%%%%%%%%%%%%%%%%%%%%%%%%%%%%%%%%%%%%%%%%%%%%%%%%%%%%%%%%%%%%%%%%%%%%%%%%%%%%%%%%%%%%%%%%%%%%%
\subsection{Finding a solution}\label{sec:selfreduce}

Let $\mathcal{I}$ be an instance of $\VCSP(\Gamma)$ and assume that the BLP
solves $\Gamma$. We will now justify this terminology by showing how to obtain an
actual assignment that optimises $\mathcal{I}$.

The basic idea is that of self-reduction: we iteratively assign labels
to the first variable of $\mathcal{I}$ and test whether the partially assigned instance
has the same optimum as $\mathcal{I}$.
When such a label is found, we proceed with the next variable.
After $n \cdot |D|$ steps, where
$n$ is the number of variables of $\mathcal{I}$, we are guaranteed
to have found an optimal assignment.
This method requires that we can find the optimum of a partially assigned
instance.
In order to do this, we need the following technical lemma which is proved in
Section~\ref{sec:generate2}.

\begin{lemmaX}\label{lemma:idem}
There exists a subset $D' \subseteq D$ such that
if $\Gamma$ admits an $m$-ary symmetric fractional polymorphism, then it admits
an $m$-ary symmetric fractional polymorphism $\omega$ such that, for all $g \in
\supp(\omega)$,
\begin{enumerate}
\item\label{item:core}
$g(x,x,\dots,x) \in D'$ for all $x \in D$;
\item\label{item:idem}
$g(x,x,\dots,x) = x$ for all $x \in D'$.
\end{enumerate}
\end{lemmaX}

The utility of Lemma~\ref{lemma:idem} can be described as follows: whenever we have
an appropriate fractional polymorphism that maps into a sub-domain $D'$
and that is idempotent on $D'$, then we can infer
that there is an optimal solution just consisting of labels from $D'$.
Such languages are made no more complex by adding
constants restricting variables to be particular labels from $D'$.
In particular, we can use BLP to solve partially assigned instances.

\begin{propositionX}
Let $\Gamma$ be an arbitrary valued constraint language
and let $\calI$ be an instance of $\VCSP(\Gamma)$.
If BLP solves $\Gamma$, then
an optimal assignment of $\calI$ can be found in polynomial time.
\end{propositionX}

\begin{proof}
Let $D' \subseteq D$ be the set in Lemma~\ref{lemma:idem}.
Let $c_d$ be the unary cost function with $c_d(x) = \infty$ for $x \neq d$ and
$c_d(d) = 0$,
and let $\Gamma_c = \Gamma \cup \{ c_d \mid d \in D' \}$.
It follows from Theorem~\ref{th:BLP:general} and
Lemma~\ref{lemma:idem}(\ref{item:idem}) that the BLP solves $\Gamma_c$.
We can now apply self-reduction to obtain an optimal assignment of
$\mathcal{I}$.
Let $x = (x_1, \dots, x_n)$ be the variables of $\mathcal{I}$.
The idea is to successively try each possible label $d \in D'$ for $x_1$ by
adding the term $c_d(x_1)$ to the sum of $\mathcal{I}$.
The modified instance is an instance of $\VCSP(\Gamma_c)$ so we can use the BLP
to obtain its optimum.
If, for some $d_1 \in D'$, the optimum of the modified instance matches that of
$\mathcal{I}$, then we know that there exists an optimal solution to
$\mathcal{I}$ in which $x_1$ is assigned $d_1$ and we can proceed with the next
variable.

We now claim that this procedure always terminates with an optimal assignment.
In particular, we must show that if the optimum of $\mathcal{I}$ is equal to the
optimum of the instance $\mathcal{I}'$ obtained by adding $\sum_{i=1}^k
c_{d_i}(x_i)$, $d_i \in D'$, $k \geq 1$, to the sum of $\mathcal{I}$, then we
can always find an optimal solution of $\mathcal{I}'$ that assigns $d_{k+1}$ to
$x_{k+1}$, for some $d_{k+1} \in D'$.
Let $s : V \to D$ be an optimal solution to $\mathcal{I}'$.
Note that $s(x_i) = d_i$ for all $i = 1, \dots, k$.
Let $\omega$ be a fractional polymorphism of $\Gamma$ satisfying
(\ref{item:core}) and (\ref{item:idem}) in Lemma~\ref{lemma:idem}.
Then, with $\mathcal{I}$ represented as in Definition~\ref{def:vcsp} on
page~\pageref{def:vcsp},
and with the notation $x^t = (x_{v(t,1)},\dots,x_{v(t,n_t)})$,
\begin{align*}
f_{\mathcal{I}'}(s(x))\ &= \ 
f_{\mathcal{I}}(s(x)) + \sum_{i=1}^k c_{d_i}(s(x_i))\\
&\geq \sum_{t \in T} \sum_{g\in \supp(\omega)} \omega(g)
f_t(g(s(x^t),\dots,s(x^t)))\\
&= \ \sum_{g\in \supp(\omega)} \omega(g)   f_{\mathcal{I}}(g[s,\dots,s](x))
\geq \ \min_{x \in D^n} f_{\mathcal{I}}(x).
\end{align*}

Since $f_{\mathcal{I}'}(s(x)) = \min_{x \in D^n} f_{\mathcal{I}}(x)$, we
conclude that $g[s,\dots,s]$ is an optimal solution to $\mathcal{I}$, for each
$g \in \supp(\omega)$. Every $g \in \supp(\omega)$ satisfies (\ref{item:idem})
in Lemma~\ref{lemma:idem}, so $g[s,\dots,s](x_i) = d_i$ for all $i = 1, \dots,
k$, and it follows that $g[s,\dots,s]$ is also an optimal solution to
$\mathcal{I}'$. Finally, $g$ satisfies (\ref{item:core}) in
Lemma~\ref{lemma:idem}, so $g[s,\dots,s](x_{k+1}) \in D'$, from which the claim
follows.
\end{proof}

%%%%%%%%%%%%%%%%%%%%%%%%%%%%%%%%%%%%%%%%%%%%%%%%%%%%%%%%%%%%%%%%%%%%%%%%%%%%%%%%%%%%%%%%%%%%%%%%%
%%%%%%%%%%%%%%%%%%%%%%%%%%%%%%%%%%%%%%%%%%%%%%%%%%%%%%%%%%%%%%%%%%%%%%%%%%%%%%%%%%%%%%%%%%%%%%%%%
%%%%%%%%%%%%%%%%%%%%%%%%%%%%%%%%%%%%%%%%%%%%%%%%%%%%%%%%%%%%%%%%%%%%%%%%%%%%%%%%%%%%%%%%%%%%%%%%%

\section{Characterisation of general-valued languages}
%\section{Proof of Theorem~\ref{th:BLP:general}}\merk{In probabilistic terminology?}

%\newcounter{mycounterTheorem}
%\setcounter{mycounterTheorem}{\value{theorem}}
%\setcounter{theorem}{\value{mycounterBLPgeneral}}

In this section we will prove the main characterisation of general-valued languages
solved by BLP.

\renewcommand{\thetheoremRESTATED}{\ref{th:BLP:general}}
\begin{theoremRESTATED}[restated]
%\begin{theoremX}[Theorem~\ref{th:BLP:general}, restated]
BLP solves $\Gamma$ if and only if $\Gamma$ admits a symmetric fractional
polymorphism of \emph{every} arity $m\ge 2$.
%\end{theoremX}
\end{theoremRESTATED}

Recall that $\mathbb M_n$ is the set of probability distributions over
labellings in $D^n$ and $\Delta=\mathbb M_1$.
For an integer $m\ge 1$ we denote by ${\mathbb M}_n^{(m)}$ the set of vectors
$\mu\in{\mathbb M}_n$ such that all components of $\mu$ are rational numbers of
the form $p/m$, where $p\in\mathbb Z$. We also denote $\Delta^{(m)}={\mathbb
M}^{(m)}_1$, and define $\Omega^{(m)}$ to be the set of mappings
$\Delta^{(m)}\rightarrow D$.

A vector $\alpha\in\Delta^{(m)}$ can be viewed as a multiset over $D$ of size $m$,
or equivalently as an element of the quotient space $D^m/\sim$ where $\sim$
is the equivalence relation defined by permutations.
Therefore, a symmetric mapping $g:D^m\rightarrow D$ can be equivalently viewed as a mapping $g:\Delta^{(m)}\rightarrow D$.
This gives a natural isomorphism between $\calO^{(m)}_{\tt sym}$ and $\Omega^{(m)}$.

A symmetric fractional polymorphism $\omega$ can thus be viewed either as a
probability distribution over $\calO^{(m)}_{\tt sym}$ or as a probability
distribution over $\Omega^{(m)}$. In the proposition below we use both views
interchangeably.

\begin{propositionX}\label{prop:sym}
Let $\omega$ be a symmetric fractional operation of arity $m$ and let $f$ be
a cost function of arity $n$. Then $\omega$ is a fractional polymorphism of
$f$ (i.e.,~\eqref{eq:FPdefinition} holds) if and only if for every
$\mu\in\mathbb M_n^{(m)}$,
\begin{equation}
\mathbb{E}_{g\sim\omega} f(g(\mu_{[1]},\ldots,\mu_{[n]})) \leq
\mathbb{E}_{x\sim\mu} f(x),
%\qquad\quad\forall \mu\in\mathbb M_n^{(m)}
\label{eq:FPdefinitionSym}
\end{equation}
%where we have denoted $f(\mu)=\sum_{x\in D^n}\mu(x)f(x)$.
where $\mu_{[i]}\in \Delta^{(m)}$ is the marginal probability of distribution
$\mu$ for the $i$th coordinate.
\end{propositionX}
\begin{proof}
The left-hand side of~\eqref{eq:FPdefinitionSym} is equal $\sum_{g\in\Omega^{(m)}}
\omega(g)f(g(\mu_{[1]},\ldots,\mu_{[n]}))$ and the right-hand side
%of~\eqref{eq:FPdefinitionSym}
is equal $\sum_{x\in D^n}\mu(x)f(x)$. The claim
thus follows from Definition~\ref{def:FP} and Remark~\ref{rem:prob}.
\end{proof}

\begin{proof}(of Theorem~\ref{th:BLP:general})
To prove Theorem~\ref{th:BLP:general}, we need to establish the following:
BLP solves $\Gamma$ if and only if for every $m\ge 2$ there exists a probability distribution $\omega$ over $\Omega^{(m)}$ that satisfies~\eqref{eq:FPdefinitionSym}.

{``$\Leftarrow$''}: Suppose that $\Gamma$ admits a symmetric fractional polymorphism of every arity.
Let $\calI$ be a $\Gamma$-instance with $n$ variables. We need to show that system~\eqref{eq:BLP} for $\calI$
has an integral minimiser.

It is well known that every LP with
rational coefficients has an optimal solution with rational coefficients (for
instance, this is a direct consequence of Fourier-Motzkin
elimination)~\cite{Schrijver86:ILP}. Therefore,
since the LP~\eqref{eq:BLP} has rational coefficients,
it has an optimal solution  $\{\alpha,\mu_t\}$ such that all variables are rational numbers of the form $p/m$ for some integers $p,m$ with $m\ge 1$.
%Let $\alpha=(\alpha_1,\ldots,\alpha_v)\in D^n$ be the corresponding part of such optimal solution.
We can assume that $m\ge 2$, otherwise the claim is trivial.

Let $\omega$ be a symmetric fractional polymorphism of $\Gamma$ of arity $m$.
For any $n$-ary cost function $f$ and $\mu\in\mathbb M_n^{(m)}$, we denote by
$f(\mu)$ the expectation $f(\mu)=\mathbb{E}_{x\sim\mu} f(x)$. 
%$=\sum_{x\in D^n}\mu(x)f(x)$.

Using~\eqref{eq:BLP} and~\eqref{eq:FPdefinitionSym},
%(and using the same notation for $f(\mu)$ as in Proposition~\ref{prop:sym}), 
we can write
%\begin{equation}
\begin{align}
\BLP(\calI)=
\sum_{t\in T}\sum_{x\in \dom f_t}\mu_t(x)f_t(x)
&=
\sum_{t\in T}f_t(\mu_t) \notag \\
&\geq
\sum_{t\in T}\sum_{g\in\Omega^{(m)}} \omega(g)f(g(\mu_{v(t,1)}),\ldots,g(\mu_{v(t,n_t)})) \notag \\
&=
\sum_{g\in\Omega^{(m)}} \omega(g)f_\calI(g(\mu)),
\label{eq:GALSFASGAS}
\end{align}
%\end{equation}
where $g$ is applied component-wise, i.e., if
$\alpha=(\alpha_1,\ldots,\alpha_n)\in[\Delta^{(m)}]^n$, then
$g(\alpha)=(g(\alpha_1),\ldots,g(\alpha_n))\in D^n$. (Note that in the second
equality in~\eqref{eq:GALSFASGAS} we have used $0\cdot\infty=0$.) Eq.~\eqref{eq:GALSFASGAS}
implies that $\BLP(\calI)\ge f_\calI(g(\mu))$ for some $g\in \supp(\omega)$,
and therefore BLP solves the instance $\calI$.

{``$\Rightarrow$''}: Let us fix $m\ge 2$, and assume that BLP solves $\Gamma$.
In this part we will use letters with a ``hat'' ($\hat\alpha$ and $\hat\mu$)
for vectors of the form $p/m$, $p\in\mathbb Z$.

First, we consider the case when $|\Gamma|$ is finite. Suppose that $\Gamma$
does  not admit a symmetric fractional polymorphism of arity $m$. Using the
notation
\begin{itemize}
\item $\Gamma^+$ is the set of tuples $(f,\hat\mu,\hat\alpha)$ such that $f$ is a function in $\Gamma$ of arity $n$,
$\hat\mu\in\mathbb M^{(m)}_n$ with $\supp(\hat\mu)\subseteq\dom f$, and
$\hat\alpha=(\hat\mu_{[1]},\ldots,\hat\mu_{[n]})\in [\Delta^{(m)}]^n$; and
\item $\Omega^{(m)}_\Gamma\subseteq\Omega^{(m)}$ is the set of mappings $g:\Delta^{(m)}\rightarrow D$ such that
$g(\hat\alpha)\in\dom f$ for all $(f,\hat\mu,\hat\alpha)\!\in\!\Gamma^+$,\!\!\!\!\!\!\!\!\!\!\!\!
%$f(g(\bx
\end{itemize}

the following system does not have a solution:

\begin{subequations}\label{eq:Farkas0}
\begin{eqnarray}
\sum_{g\in\Omega^{(m)}_\Gamma}\omega(g)f(g(\hat\alpha))&\le& f(\hat\mu)\qquad \forall (f,\hat\mu,\hat\alpha)\in\Gamma^+ \\ % .\supp(\alpha)\subseteq dom f \\
\sum_{g\in\Omega^{(m)}_\Gamma}\omega(g) & = & 1 \\
\omega(g)&\ge& 0 \hspace{38pt} \forall g\in\Omega^{(m)}_\Gamma
\end{eqnarray}
\end{subequations}

Since system~\eqref{eq:Farkas0} is infeasible, by Farkas' lemma~\cite{Schrijver86:ILP} the following system has a solution:
\begin{subequations}\label{eq:Farkas1}
\begin{eqnarray}
\sum_{(f,\hat\mu,\hat\alpha)\in\Gamma^+} y(f,\hat\mu,\hat\alpha)  f(\hat\mu) + z &<& 0  \\
\sum_{(f,\hat\mu,\hat\alpha)\in\Gamma^+} y(f,\hat\mu,\hat\alpha) f(g(\hat\alpha)) + z &\ge& 0 \qquad \forall g\in\Omega^{(m)}_\Gamma \\
y(f,\hat\mu,\hat\alpha) & \ge & 0 \qquad\forall (f,\hat\mu,\hat\alpha)\in\Gamma^+
\end{eqnarray}
\end{subequations}
Eliminating $z$ gives
\begin{subequations}\label{eq:Farkas2}
\begin{eqnarray}
\sum_{(f,\hat\mu,\hat\alpha)\in\Gamma^+} y(f,\hat\mu,\hat\alpha) f(g(\hat\alpha))  &>& \sum_{(f,\hat\mu,\hat\alpha)\in\Gamma^+} y(f,\hat\mu,\hat\alpha) f(\hat\mu) \qquad \forall g\in\Omega^{(m)}_\Gamma \label{eq:Farkas2:a} \\
y(f,\hat\mu,\hat\alpha) & \ge & 0 \hspace{117pt} \forall (f,\hat\mu,\hat\alpha)\in\Gamma^+\qquad
\end{eqnarray}
\end{subequations}
We claim that vector $y$ in~\eqref{eq:Farkas2} can be chosen to be integer-valued and strictly positive. To see this, observe
that if $y$ is a feasible solution then so is any vector $y'$ with $y'(f,\hat\mu,\hat\alpha)\in[C y(f,\hat\mu,\hat\alpha),C y(f,\hat\mu,\hat\alpha)+1]$ for $(f,\hat\mu,\hat\alpha)\in\Gamma^+$, 
for some sufficiently large constant $C$ (namely, $C>\frac{2}{\epsilon}\max_{f\in\Gamma,x\in\dom f} |f(x)| $ where $\epsilon>0$ is the minimum difference between the left-hand side and the right-hand side 
in~\eqref{eq:Farkas2:a}).

Let us construct an instance $\calI$ with variables $V=\Delta^{(m)}$ and the function
\begin{eqnarray}
f_\calI(x)=\sum_{(f,\hat\mu,\hat\alpha)\in\Gamma^+,\hat\alpha=(\hat\alpha_1,\ldots,\hat\alpha_n)}y(f,\hat\mu,\hat\alpha)f(x_{\hat\alpha_1},\ldots,x_{\hat\alpha_n}).
\end{eqnarray}
This can be viewed as a $\Gamma$-instance, if we simulate 
 the multiplication of $y(f,\hat\mu,\hat\alpha)$ and $f$ by repeating the latter term $y(f,\hat\mu,\hat\alpha)$ times.

Consider a mapping $g:\Delta^{(m)}\rightarrow D$. Since $V=\Delta^{(m)}$, such $g$ is
a valid labelling for the instance $\calI$, so we can evaluate $f_\calI(g)$.
If $g\in\Omega^{(m)}\setminus\Omega^{(m)}_\Gamma$ then $f_\calI(g)=\infty$ (by the definition of $\Omega^{(m)}_\Gamma$),
and if $g\in \Omega^{(m)}_\Gamma$ then $f_\calI(g)$ equals the left-hand side of eq.~\eqref{eq:Farkas2:a}.
Thus, eq.~\eqref{eq:Farkas2:a} gives
\begin{equation}
\min_{g:V\rightarrow D} f_\calI(g)  > \sum_{(f,\hat\mu,\hat\alpha)\in\Gamma^+} y(f,\hat\mu,\hat\alpha) f(\hat\mu).  \label{eq:Farkas2:a'}
\end{equation}
The left-hand side of eq.~\eqref{eq:Farkas2:a'} is the optimal value of the instance $\calI$.
We claim that 
\begin{equation}
\sum_{(f,\hat\mu,\hat\alpha)\in\Gamma^+} y(f,\hat\mu,\hat\alpha) f(\hat\mu) \ge \BLP(\calI)
\label{eq:Farkas2:last}
\end{equation}
and so $\min_{g} f_\calI(g)>\BLP(\calI)$, which contradicts the assumption that BLP solves $\Gamma$.
To prove eq.~\eqref{eq:Farkas2:last}, it suffices to specify a feasible vector of the BLP relaxation of $\calI$
(given by eq.~\eqref{eq:BLP})
whose value equals the left-hand side of \eqref{eq:Farkas2:last}. Such vector is constructed as follows:
$\alpha_{v}=v$ for all $v\in V=\Delta^{(m)}$
and $\mu_t=\hat \mu$ for all $t=(f,\hat\mu,\hat\alpha)\in\Gamma^+$.

We showed that if $\Gamma$ is finite then it admits a symmetric fractional
polymorphism of arity $m$.
Now suppose that $\Gamma$ is infinite but countable:
$\Gamma=\{f_1,f_2,\ldots\}$.
For any integer $r\ge 1$, denote by $\Gamma_r=\{f_1,\ldots,f_r\}$.
As we just showed, $\Gamma_r$ admits some symmetric fractional polymorphism
$\omega_r$ of arity $m$.
The space of symmetric fractional polymorphisms is a compact subset of $\mathbb
R^{|\Omega^{(m)}|}$,
therefore the sequence $\omega_1,\omega_2,\ldots$ has a limit vector $\omega$.
Using standard continuity arguments, we conclude that $\omega$ is a symmetric
fractional polymorphism of $\Gamma$.
This concludes the proof.
\end{proof}

%%%%%%%%%%%%%%%%%%%%%%%%%%%%%%%%%%%%%%%%%%%%%%%%%%%%%%%%%%%%%%%%%%%%%%%%%%%%%%%%%%%%%%%%%%%%%%%%%
%%%%%%%%%%%%%%%%%%%%%%%%%%%%%%%%%%%%%%%%%%%%%%%%%%%%%%%%%%%%%%%%%%%%%%%%%%%%%%%%%%%%%%%%%%%%%%%%%
%%%%%%%%%%%%%%%%%%%%%%%%%%%%%%%%%%%%%%%%%%%%%%%%%%%%%%%%%%%%%%%%%%%%%%%%%%%%%%%%%%%%%%%%%%%%%%%%%

\section{Constructing new fractional polymorphisms}\label{sec:construction}
In this section we introduce a generic procedure for constructing new fractional polymorphisms of a language $\Gamma$
from existing ones. This procedure will be used in the proofs of Theorems~\ref{th:BLP:generate}
and \ref{thm:finite:gen}, and Lemma~\ref{lemma:idem}.

We start with a motivating example to illustrate techniques that we will use.
\begin{exampleX}\label{ex:motivate}\examplefont
Let us consider all idempotent binary operations $D^2\rightarrow D$ for the set of labels $D=\{0,1\}$. 
In total, there are four such operations: $g_{00},g_{01},g_{10},g_{11}$ where $g_{ab}$ for $a,b\in D$ denotes the operation with
\begin{eqnarray*}
g_{ab}(0,0)=0, \quad
g_{ab}(0,1)=a, \quad
g_{ab}(1,0)=b, \quad
g_{ab}(1,1)=1.
\end{eqnarray*}
%Note that $g_{ab}$ and $g_{ba}$ are the projections to the first and second coordinate respectively,
%and $g_{aa},g_{bb}$ are symmetric operations.
The operations $g_{01}$ and $g_{10}$ are the two binary projections $e^{(2)}_1$ and
$e^{(2)}_2$ to the first and second coordinate respectively.
The operations $g_{00}$ and $g_{11}$ are the $\min$ and $\max$ operations
with respect to the natural order on $D$.

Suppose that $\Gamma$ admits a binary fractional polymorphism $\omega=\frac{1}{3}\chi_{g_{00}}+\frac{2}{3}\chi_{g_{01}}$
(so that one operation in $\supp(\omega)$ is symmetric and the other is not), and we want to prove that $\Gamma$ admits
a symmetric binary fractional polymorphism. For  a function $f\in\Gamma$ and labellings $x,y\in\dom f$ 
we can write 
%
%\begin{subequations}
\begin{eqnarray}
\frac{f(x)+f(y)}{2}
&\ =\ &\frac{1}{2}\frac{f(x)+f(y)}{2} + \frac{1}{2}\frac{f(y)+f(x)}{2} \nonumber \\
&\ge& \frac{1}{2}\sum_{g\in\supp(\omega)}\omega(g) f(g(x,y)) + \frac{1}{2}\sum_{g\in\supp(\omega)}\omega(g) f(g(y,x)) \nonumber \\
&=& \frac{1}{3} f(g_{00}(x,y)) + \frac{2}{3}\frac{f(x) + f(y)}{2} \label{eq:example:inequality1}
%&\ge& \frac{1}{3} f(g_{00}(x,y)) + \frac{2}{3}\left[
%\frac{1}{3} f(g_{00}(x,y)) + \frac{2}{3}\frac{f(x) + f(y)}{2}
%\right] \nonumber \\
%&=& \frac{5}{9} f(g_{00}(x,y)) + \frac{4}{9}\frac{f(x) + f(y)}{2} \\
%&\ge&\ldots\nonumber
\end{eqnarray}
%\end{subequations}
%
This inequality means that vector $\rho_1=\frac{1}{3}\chi_{g_{00}}+\frac{2}{3}\frac{\chi_{g_{01}} + \chi_{g_{10}}}{2}$
is a fractional polymorphism of $\Gamma$. 
This demonstrates how we can derive new fractional polymorphisms of a language from existing ones by taking superpositions.

To obtain a symmetric fractional polymorphism of $\Gamma$, we can now use two strategies:
\begin{itemize}
\item[(i)] Cancel terms in~\eqref{eq:example:inequality1}, obtaining inequality 
$\frac{1}{3}\frac{f(x)+f(y)}{2}\ge \frac{1}{3} f(g_{00}(x,y))$. This inequality means that
vector $\chi_{g_{00}}$ is a fractional polymorphism of $\Gamma$.
\item[(ii)] Take the last term $\frac{2}{3}\frac{f(x) + f(y)}{2}$ in~\eqref{eq:example:inequality1}
and apply $\omega$ to it again. This gives a new inequality corresponding to a fractional polymorphism 
$\rho_2=\frac{5}{9}\chi_{g_{00}}+\frac{4}{9}\frac{\chi_{g_{01}} + \chi_{g_{10}}}{2}$.
By repeating this process we obtain a sequence of vectors $\rho_1,\rho_2,\rho_3,\ldots$
that are fractional polymorphisms of $\Gamma$. The weight of operation $g_{00}$ in $\rho_i$
tends to 1 as $i$ tends to infinity; thus, by taking the limit we can prove that vector $\chi_{g_{00}}$ is a fractional
polymorphism of $\Gamma$.
\end{itemize}

\end{exampleX}

Observe that in the construction in Example~\ref{ex:motivate} operations $g_{01}$ and $g_{10}$ always had the same weight in $\rho_i$.
Therefore, we were working with a subset of all possible fractional polymorphisms. 
This example motivates definitions given below.
%
%Suppose language $\Gamma$ admi
%In Theorem~\ref{th:BLP:generate} we need to prove the existence of a symmetric fractional oper
%
%{\small
%\begin{IEEEeqnarray*}{rCll}
%\frac{f(x^1)+f(x^2)}{2}
%&\le& \frac{1}{3} f(g_{aa}(x^1,x^2)) + \frac{2}{3}\frac{f(g_{ab}(x^1,x^2)) + f(g_{ba}(x^1,x^2))}{2} \\
%& & \mbox{// denote } y^1=g_{ab}(x^1,x^2),y^2=g_{ba}(x^1,x^2) \\
%&\le& \frac{1}{3} f(g_{aa}(x^1,x^2)) + \frac{2}{3}\left[
%\frac{1}{3} f(g_{aa}(y^1,y^2)) + \frac{2}{3}\frac{f(g_{ab}(y^1,y^2)) + f(g_{ba}(y^1,y^2))}{2}
%\right] 
%\end{IEEEeqnarray*}
%}

\subsection{Generalised fractional polymorphisms}\label{sec:gfPol}

The construction of new fractional polymorphisms is based on the idea of
grouping operations in $\calO^{(m)}$ together into what we will call \emph{collections}
and working with fractional operations that assign the same weight to every operation
in a collection.
We will consider two types of collections: \emph{ordered} and \emph{unordered}.
Ordered collections are finite sequences of operations from $\calO^{(m)}$ and
unordered collections are subsets of $\calO^{(m)}$.

Let $\mathbb G$ be a fixed set of collections.
We will always assume that all collections in $\mathbb G$ are of the same type,
i.e., either ordered or unordered. (In Example~\ref{ex:motivate} above we would use $\mathbb G=\{\{g_{00}\},\{g_{11}\},\{g_{01},g_{10}\}\}$.)
For a collection ${\bf g} \in \mathbb G$, we let $|{\bf g}|$ denote its \emph{size}, i.e.,
the cardinality of the set, or the length of the sequence, depending on its type.
We write $\sum_{g \in{\bf g}}$ to denote a sum over all components of the (ordered or unordered) collection ${\bf g}$;
one has e.g.\  $\sum_{g \in{\bf g}}1=|{\bf g}|$.

For a given $\mathbb G$, we define a \emph{generalised fractional polymorphism}
$\rho$ of a language $\Gamma$ as a probability distribution over $\mathbb G$ such that,
for every cost function $f \in \Gamma$,
\begin{equation}\label{eq:gfpol}
\sum_{{\bf g}\in\mathbb G} \rho({\bf g})\sum_{g\in{\bf g}}\frac{1}{|{\bf
g}|}f(g(x^1,\ldots,x^m)) \leq f^m(x^1,\ldots,x^m), \qquad\forall
x^1,\ldots,x^m\in \dom f.
\end{equation}

To simplify notation, we will write~\eqref{eq:gfpol} as
\begin{equation}\label{eq:gfpolalt}
\sum_{{\bf g}\in\mathbb G} \rho({\bf g})f^{|\bf g|}({\bf g}(x^1,\dots,x^m)) \leq
 f^m(x^1,\dots,x^m), \qquad\forall x^1,\dots,x^m\in \dom f,
\end{equation}
where, for an unordered collection ${\bf g}=\{g_1,\dots,g_k\}$ of $m$-ary
operations, the application ${\bf g}(x^1,\dots,x^m)$ denotes the set of
labellings $\{g_1(x^1,\dots,x^m),\dots,g_k(x^1,\dots,x^m)\}$, and we define
$f^k(\{y^1,\dots,y^k\})\eqdef f^k(y^1,\dots,y^k)$ for labellings
$y^1,\dots,y^k\in D^n$. Similarly, for an ordered collection ${\bf
g}=(g_1,\dots,g_k)$ of $m$-ary operations, the application ${\bf
g}(x^1,\dots,x^m)$ denotes the sequence of labellings
\[
(g_1(x^1,\dots,x^m),\dots,g_k(x^1,\dots,x^m)).\]

As for fractional operations, we
define the $\supp(\rho)$ as the set of collections ${\bf g}$ for which $\rho({\bf g}) > 0$,
and the vector $\chi_{\bf g}$ as the vector that assigns weight $1$ to the collection ${\bf g}$
and $0$ to all other collections.

Note that $\rho$ is a generalised fractional polymorphism of $\Gamma$ if and only if
$\omega=\sum_{{\bf g}\in{\mathbb G}} \rho({\bf g}) \sum_{g\in{\bf g}} \frac{1}{|{\bf g}|} \chi_g$ is an $m$-ary fractional polymorphism of $\Gamma$.

\paragraph{Terminology for ordered collections} Observe that an ordered collection ${\bf g}=(g_1,\ldots,g_k)$ is a mapping $D^m\rightarrow
D^k$.
Denote the set of such mappings by $\calO^{(m\rightarrow k)}$. If $\mathbb G$ is a subset of $\calO^{(m\rightarrow k)}$,
then a generalised fractional polymorphism $\rho$ over $\mathbb G$ will be called a 
\emph{generalised fractional polymorphism of arity $m\rightarrow k$}.
We can identify fractional polymorphisms of arity $m$ with generalised fractional polymorphisms of arity $m\rightarrow 1$.

\subsection{Constructing generalised fractional polymorphisms}
Our goal will be to construct a generalised fractional polymorphism $\rho$
such that all collections ${\bf g}\in\supp(\rho)$ satisfy some desired property.
Let $\mathbb G^\ast\subseteq\mathbb G$ be the set of such ``good'' collections.
We consider an {\em expansion operator} {\sf Exp} that takes
a collection ${\bf g}\in\mathbb G$ and produces a probability distribution $\rho$ over $\mathbb G$.
We say that  {\sf Exp} is \emph{valid} for a language $\Gamma$ if,
for any $f\in\Gamma$
and any ${\bf g}\in\mathbb G$, the probability distribution $\rho={\sf Exp}(\bf g)$ satisfies
\begin{equation}\label{eq:valid}
\sum_{{\bf h}\in\mathbb G} \rho({\bf h})f^{|\bf h|}({\bf h}(x^1,\ldots,x^m))
\le f^{|{\bf g}|}({\bf g}(x^1,\ldots,x^m)), \qquad\forall x^1,\ldots,x^m\in \dom f.
\end{equation}
We say that the operator {\sf Exp} is \emph{non-vanishing} (with respect to the pair $(\mathbb G,\mathbb G^\ast)$)
if, for any ${\bf g}\in\mathbb G$, there exists a sequence of collections ${\bf g}_0,{\bf g}_1,\ldots,{\bf g}_r$
such that ${\bf g}_0={\bf g}$, ${\bf g}_{i+1}\in\supp({\sf Exp}({\bf g}_i))$ for $i=0,\ldots,r-1$,
 and ${\bf g}_r\in\mathbb G^\ast$.
 
The main result of this section is the following.

\begin{lemmaX}[``Expansion Lemma'']
Let {\sf Exp} be an expansion operator which is valid for the
language $\Gamma$ and non-vanishing with respect to $(\mathbb G,\mathbb G^\ast)$.
If $\Gamma$ admits a generalised fractional polymorphism $\rho$ 
with $\supp(\rho) \subseteq \mathbb G$, then it also admits a generalised fractional polymorphism
$\rho^\ast$ with $\supp(\rho^\ast) \subseteq \mathbb G^\ast$.
\label{lemma:expansion}
\end{lemmaX}

The example below demonstrates how this lemma can be used.
%\begin{exampleX}\label{ex:motivate:revisit}\examplefont
\renewcommand{\theexampleRESTATED}{\ref{ex:motivate}}
\begin{exampleRESTATED}[revisited]
Recall that in this example we have a language $\Gamma$ over $D=\{0,1\}$
that admits a fractional polymorphism $\omega=\frac{1}{3}\chi_{g_{00}} + \frac{2}{3}\chi_{g_{01}}$.
Let us show how we can use the Expansion Lemma to derive the fact
that $\Gamma$ admits a symmetric fractional polymorphism.

Let $\mathbb G=\{\{g_{00}\},\{g_{11}\},\{g_{01},g_{10}\}\}$
and $\mathbb G^\ast=\{\{g_{00}\},\{g_{11}\}\}$. Note that vector $\rho=\chi_{\{g_{01},g_{10}\}}$
is a generalised fractional polymorphism of $\Gamma$.
Define expansion operator {\sf Exp} as follows:
\begin{eqnarray*}
{\sf Exp}(\{g_{00}\})&=&\sum_{g\in\supp(\omega)}\omega(g)\chi_{\{g[g_{00},g_{00}]\}}=\chi_{\{g_{00}\}} \\
{\sf Exp}(\{g_{11}\})&=&\sum_{g\in\supp(\omega)}\omega(g)\chi_{\{g[g_{11},g_{11}]\}}=\chi_{\{g_{11}\}} \\
{\sf Exp}(\{g_{01},g_{10}\})&=&\sum_{g\in\supp(\omega)}\omega(g)\chi_{\{g[g_{01},g_{10}],g[g_{10},g_{01}]\}}
=\frac{1}{3}\chi_{\{g_{00}\}} + \frac{2}{3} \chi_{\{g_{01},g_{10}\}}
\end{eqnarray*}
It can be verified {\sf Exp}
 is valid for $\Gamma$ and non-vanishing with respect to $(\mathbb G,\mathbb G^\ast)$;
we leave this to the reader. By the Expansion Lemma, $\Gamma$ admits a generalised fractional polymorphism
$\rho^\ast$ with $\supp(\rho^\ast)\subseteq\mathbb G^\ast$, i.e.\ $\rho^\ast=\rho^\ast_{00}\chi_{\{g_{00}\}}+\rho^\ast_{11}\chi_{\{g_{11}\}}$
for some $\rho^\ast_{00},\rho^\ast_{11}\ge 0$ with $\rho^\ast_{00}+\rho^\ast_{11}=1$.
This means that vector $\rho^\ast_{00}\chi_{g_{00}}+\rho^\ast_{11}\chi_{g_{11}}$ is  a (symmetric) fractional polymorphism of $\Gamma$.

Note that in the original Example~\ref{ex:motivate} we claimed a stronger property,
namely that $\Gamma$ admits a fractional polymorphism $\chi_{g_{00}}$.
This can be established by taking $\mathbb G=\{\{g_{00}\},\{g_{01},g_{10}\}\}$ and
$\mathbb G^\ast=\{\{g_{00}\}\}$; the expansion operator for these collections is defined as above.
\end{exampleRESTATED}

The lemma will be used for constructing the desired fractional polymorphisms
in Theorem~\ref{th:BLP:generate} and Lemma~\ref{lemma:idem}; their
proofs consist of exhibiting an appropriate pair $(\mathbb G,\mathbb G^\ast)$ and
an expansion operator ${\sf Exp}$. 
We will also exploit the lemma in the proof of Theorem~\ref{thm:finite:gen}. However, on its own it will not be
enough, and we will need an additional method for constructing symmetric fractional polymorphisms.
Considering Example~\ref{example:STPcycle}, this should not come as a surprise: 
Theorem~\ref{thm:finite:gen} holds only for finite-valued languages,
while Lemma~\ref{lemma:expansion} is valid for all general-valued languages.

We will now give two proofs of Lemma~\ref{lemma:expansion}.
The first one is constructive: it shows how to obtain $\rho^\ast$ from $\rho$
using a finite number of steps. The second one is shorter but non-constructive.
Roughly speaking, the constructive proof uses strategy (i) from Example~\ref{ex:motivate},
while the non-constructive one uses strategy (ii).

\subsection{Constructive proof of the Expansion Lemma} %% ~\ref{lemma:expansion}}

In this subsection, we give a constructive proof of the Expansion Lemma %Lemma~\ref{lemma:expansion}
based on a node-weighted tree generated by repeated applications of the expansion operator.

\begin{proof}(of Lemma~\ref{lemma:expansion})
The proof goes via an explicit construction of a node-weighted tree.
Each node of the tree contains a collection ${\bf g} \in {\mathbb G}$, and we will
use ${\bf g}$ to denote both the node and the collection of operations it contains.
We say that two nodes are \emph{equal as collections} if the collections they contain are the same.
Each node also carries a (strictly) positive weight, denoted by $w({\bf g})$.
We say that a node \emph{${\bf g}$ is covered (by ${\bf h}$)},
and that ${\bf h}$ is a \emph{covering node (of ${\bf g}$)} if 
${\bf g}$ is a descendant of ${\bf h}$ and ${\bf g}$ and ${\bf h}$ are equal as collections.
We say that ${\bf h}$ is a \emph{minimal} covering node if no descendant of ${\bf h}$ is a covering node.

The root of the tree will be denoted by ${\bf p}$.
If the collections in $\mathbb G$ are unordered, then we let
${\bf p} = \{e^{(m)}_1, \dots, e^{(m)}_m\}$ be the set of $m$-ary projections.
If the collections are ordered,
then we let ${\bf p} = (e^{(m)}_1, \dots, e^{(m)}_m)$.
If ${\bf p}$ is not in $\mathbb G$,
then we can augment $\mathbb G$ with ${\bf p}$ and
define ${\sf Exp}({\bf p})$ to be $\rho$.
In both cases, 
since $\Gamma$ admits $\rho$,
the augmented expansion operator remains non-vanishing and valid for $\Gamma$.

The construction is performed in two steps. In the first step, the \emph{expansion},
a tree with root ${\bf p}$ is constructed. 
At the end of this step, the tree will have leaves that are either in
${\mathbb G}^\ast$ or that are covered.
From the construction, it will immediately follow that the leaves induce a
generalised fractional polymorphism of $\Gamma$.
In the second step, the \emph{pruning}, certain parts of the tree will be cut down
to remove all leaves that are not in ${\mathbb G}^\ast$.
We then prove that the set of leaves remaining 
after the pruning step still induces a generalised fractional polymorphism of $\Gamma$.

The expansion step is carried out as follows.
\begin{itemize}
\item
While there exists a leaf ${\bf g}$ in the tree that is in ${\mathbb G} \setminus {\mathbb G}^\ast$ 
and not covered,
add the set $\supp({\sf Exp}({\bf g}))$ as children to ${\bf g}$ with the weight
of each added node ${\bf h}$
given by $w(\bf g) \cdot {\sf Exp}({\bf g})({\bf h})$.
\end{itemize}

First, we argue that the expansion step terminates after a finite number of applications
of the expansion operator.
Assume to the contrary that the expansion generates an infinite tree.
Since there is a finite number of $m$-ary operations, $|{\mathbb G}|$ is finite, and hence
it follows that the tree has an infinite path ${\bf g}_0, {\bf g}_1, \dots$,
descending from the root.
Therefore, there exist $i < j$ such that ${\bf g}_i = {\bf g}_j$ as collections.
If ${\bf g}_j \in {\mathbb G}^\ast$, then ${\bf g}_i$ would never have been
expanded, so we may assume that ${\bf g}_j \not\in {\mathbb G}^\ast$.
But then ${\bf g}_j$ is covered by ${\bf g}_i$, hence ${\bf g}_j$ would never have been
expanded.
We have reached a contradiction and it follows that the tree generated by the expansion is finite.

Let $T$ be the tree generated by the expansion step.
For a node ${\bf g}$ in $T$, let $L({\bf g})$ be the set of leaves of the sub-tree rooted at ${\bf g}$
and define the vector $\nu_{\bf g} \in \R_{\geq 0}^{\mathbb G}$ by
$\nu_{\bf g} = \sum_{{\bf h} \in L({\bf g})} (w({\bf h})/w({\bf g})) \chi_{\bf h}$.
We claim that the following three properties are satisfied:

\begin{enumerate}[(a)]
\item\label{inv:new}
For every node ${\bf g}$ in $T$, the vector $\nu_{\bf g}$ is a probability distribution on $\mathbb G$
such that for every function $f\in\Gamma$, and all
  tuples $x^1, \dots, x^m \in \dom f$,
\begin{equation}\label{eq:g}
\sum_{{\bf h}\in \mathbb G} \nu_{\bf g}({\bf h})f^{|\bf h|}({\bf h}(x^1,\ldots,x^m))
\le f^{|{\bf g}|}({\bf g}(x^1,\ldots,x^m)).
\end{equation}
\item\label{inv:cov}
Every leaf in $T$ is either a member of ${\mathbb G}^\ast$ or covered.
\item\label{inv:leaves}
Every sub-tree of $T$ contains a leaf that is not covered by a node in that sub-tree.
\end{enumerate}

Property (a) follows by repeated application of (\ref{eq:valid}), with
special care taken to the case $\nu_{\bf p}$, where the fact that $\rho$ satisfies
(\ref{eq:gfpolalt}) is also used.
Property (b) holds trivially when the expansion step terminates.
To see that  (c) holds,
assume to the contrary that, after the expansion step,
there is a sub-tree rooted at ${\bf g}$ for which every leaf is in
$\mathbb G \setminus \mathbb G^\ast$.
Pick any node ${\bf g}_0$ in this sub-tree, and for $i \geq 0$, arbitrarily
pick ${\bf g}_{i+1} \in \supp({\sf Exp}({\bf g}_i))$.
Then, each ${\bf g}_i$ contains a collection that already exists in the sub-tree,
hence the sequence ${\bf g}_0, {\bf g}_1, \dots$ never encounters a collection in $\mathbb G^\ast$.
This is a contradiction since ${\sf Exp}$ is assumed to be non-vanishing.
Hence, (\ref{inv:leaves}) holds after the expansion step.

The pruning step is carried out as follows.
\begin{itemize}
\item
While there exists a covered leaf in the tree,
pick a minimal covering node ${\bf g}$ and let $T_{\bf g}$ be the
sub-tree rooted at ${\bf g}$.
Write $\nu_{\bf g}$ as
$\nu_{\bf g} = (1-\kappa)\cdot \chi_{\bf g} + \nu_\bot$, where
$1-\kappa = \nu_{\bf g}({\bf g})$ so that $\nu_\bot({\bf g}) = 0$.
Remove all nodes below ${\bf g}$ from the tree and, for each collection
${\bf h} \in \mathbb G$ such that $\nu_\bot({\bf h}) > 0$, %%is different from ${\bf g}$ as a set,
add a new leaf ${\bf h}'$ as a child to ${\bf g}$ containing the collection ${\bf h}$ and with
weight
$w({\bf h}') = w({\bf g}) \frac{1}{\kappa} \nu_\bot({\bf h})$.
\end{itemize}

We refer to each choice of a minimal covering node ${\bf g}$ and the subsequent 
restructuring of $T_{\bf g}$ as a \emph{round} of pruning.
Below, we prove that the properties (\ref{inv:new}--\ref{inv:leaves}) are invariants
that hold before and after each round of pruning.
This has the following consequences.
Since (\ref{inv:leaves}) holds before each round,
some leaf ${\bf h}$ of ${\bf g}$ is different from ${\bf g}$ as a collection,
and since $w({\bf h})$ is strictly positive,
$\kappa > 0$.
Therefore the new weights in the pruning step are defined.
Each round of pruning decreases the size of the tree by at least one, so the pruning step
eventually terminates.
After the pruning step has terminated, let $\rho^\ast$ be $\nu_{\bf p}$.
By (\ref{inv:new}), $\rho^\ast$ is a generalised
fractional polymorphism of $\Gamma$,
and by (\ref{inv:cov}), every leaf must contain a collection in $\mathbb{G}^\ast$,
so $\supp(\rho^\ast) \subseteq \mathbb G^\ast$.

We finish the proof by showing that (\ref{inv:new}--\ref{inv:leaves}) hold after a
round of pruning that picks a minimal covering node ${\bf g}$, 
assuming that they held before the round.
By (\ref{inv:new}) and noting that $\sum_{{\bf h} \in \mathbb G} \chi_{\bf g}({\bf h}) f^{|{\bf h}|}({\bf h}(x^1,\dots,x^m)) = f^{|{\bf g}|}({\bf g}(x^1,\dots,x^m))$, we have

\begin{align}\label{eq:chiplusnu}
\sum_{{\bf h} \in \mathbb G} \nu_{\bf g}({\bf h}) f^{|{\bf g}|}({\bf g}(x^1, \dots, x^m)) &=
\sum_{{\bf h} \in \mathbb G} ((1-\kappa) \cdot \chi_{\bf g}+\nu_\bot)({\bf h}) f^{|\bf h|}({\bf h}(x^1, \dots, x^m)) \notag \\
&\geq
\sum_{{\bf h} \in \mathbb G} ((1-\kappa) \cdot \nu_{\bf g}+ \nu_\bot)({\bf h}) f^{|\bf h|}({\bf h}(x^1, \dots, x^m)),
\end{align}
for all $f \in \Gamma$ and $x^1,\dots,x^m\in\dom f$.
Since $\kappa > 0$, inequality (\ref{eq:chiplusnu}) is equivalent to
\begin{equation}\label{eq:pruning-ineq}
\sum_{{\bf h} \in \mathbb G} \nu_{\bf g}({\bf h}) f^{|\bf h|}({\bf h}(x^1, \dots, x^m)) \geq
  \sum_{{\bf h} \in \mathbb G} \frac{1}{\kappa} \nu_\bot({\bf h}) f^{|{\bf h}|}({\bf h}(x^1, \dots, x^m)).
  %%f^{|{\bf g}|}({\bf g}(x^1, \dots, x^m)).
\end{equation}

The round of pruning only affects the vector $\nu_{{\bf g}'}$ for nodes ${\bf g'}$ that lie on the path from
the root to ${\bf g}$.
Let ${\bf g}'$ be such a node and
let $\nu'_{{\bf g}'}$ be the altered function after the round.
Let $C = w({\bf g})/w({\bf g}')$.
Then, $\nu'_{{\bf g}'} = \nu_{{\bf g}'} - C \nu_{\bf g} + C \frac{1}{\kappa}\nu_{\bot}$
and $\nu'_{{\bf g}'}$ can easily be verified to be a probability distribution on $\mathbb G$, and

\begin{align}
f^{|{\bf g}'|}({\bf g}'(x^1, \dots, x^m)) &
\geq 
\sum_{{\bf h} \in \mathbb G} (\nu_{{\bf g}'}-C \nu_{{\bf g}}+C\nu_{{\bf g}})({\bf h}) f^{|{\bf h}|}({\bf h}(x^1,\dots,x^m)) \notag \\
& \geq
\sum_{{\bf h} \in \mathbb G} (\nu_{{\bf g}'}-C\nu_{{\bf g}}+C\frac{1}{\kappa}\nu_{\bot})({\bf h}) f^{|{\bf h}|}({\bf h}(x^1,\dots,x^m)) \notag \\
& = 
\sum_{{\bf h} \in \mathbb G} \nu'_{{\bf g}'}({\bf h}) f^{|{\bf h}|}({\bf h}(x^1,\dots,x^m)),
\end{align}
so (\ref{inv:new}) holds after the round.
%
%%\begin{align}
%%\sum_{{\bf h} \in \mathbb G} w({\bf g}) \nu_{\bf g}({\bf h}) f^{|{\bf h}|}(x^1,\dots,x^m) &=
%%\sum_{{\bf h} \in \mathbb G} w({\bf g}) ((1-\kappa) \chi_{\bf g}+\nu_{\bot})({\bf h}) f^{|{\bf h}|}(x^1,\dots,x^m) \\ &\geq
%%\sum_{{\bf h} \in \mathbb G} w({\bf g}) \left(\frac{1-\kappa}{\kappa} + 1\right)\nu_{\bot}({\bf h}) f^{|{\bf h}|}(x^1,\dots,x^m).
%%\end{align}

Before the round of pruning, using (\ref{inv:cov}) and the fact that ${\bf g}$ is a minimal covering node,
every leaf in $T_{\bf g}$ is either equal to ${\bf g}$ as collections, 
is a member of ${\mathbb G}^\ast$, or is covered by a node above ${\bf g}$ in the tree.
Therefore, every new child ${\bf h}'$ added to ${\bf g}$ in the pruning is either a member
of ${\mathbb G}^\ast$ or is still covered by some node above ${\bf g}$, so (\ref{inv:cov}) holds
after the round.

The round of pruning only affects the leaves of the sub-trees $T_{{\bf g}'}$ that contain ${\bf g}$.
Before the round,
by (\ref{inv:leaves}), the sub-tree $T_{{\bf g}'}$ either contains a leaf in ${\mathbb G}^\ast$
or a leaf that is covered by a node above ${{\bf g}'}$ in the tree.
If such a leaf ${\bf h}$ is also a leaf of $T_{\bf g}$, and therefore potentially altered in
the round of pruning,
then ${\bf h}$ is different from ${\bf g}$ as collections, hence ${\bf g}$ will have a new child ${\bf h}'$
with the same property.
It follows that (\ref{inv:leaves}) holds after the round.
%Finally, only leaves that are equal to ${\bf g}$ as sets have been removed and these were all covered by ${\bf g}$. Therefore, in the modified tree, Invariant (\ref{inv:leaves}) holds for all sub-trees containing ${\bf g}$ (including $T'_{\bf g}$).
\end{proof}

\subsection{Non-constructive proof of the Expansion Lemma} %% ~\ref{lemma:expansion}}
We now give a non-constructive proof of the Expansion Lemma.
\begin{proof}(of Lemma~\ref{lemma:expansion})
Let $\Omega$ be the set of generalised fractional polymorphisms $\rho$
of $\Gamma$ with $\supp(\rho)\subseteq \mathbb G$; it is non-empty by the assumption of Lemma~\ref{lemma:expansion}.
Let us pick $\rho^\ast\in\Omega$ with the maximum value of $\rho^\ast(\mathbb G^\ast)\eqdef \sum_{{\bf g}\in\mathbb G^\ast}\rho({\bf g})$.
(Clearly, $\Omega$ is a compact set which is a subset of $\mathbb R^{|\mathbb G|}$, so the maximum is attained by some vector in $\Omega$).
We claim that $\supp(\rho^\ast)\subseteq\mathbb G^\ast$. Indeed, suppose that ${\bf g}_0\notin \mathbb G^\ast$ for some ${\bf g}_0\in\supp(\rho^\ast)$.
Since ${\sf Exp}$ is non-vanishing,
there exists a sequence ${\bf g}_0,{\bf g}_1,\ldots,{\bf g}_r$
such that  ${\bf g}_{i+1}\in\supp({\sf Exp}({\bf g}_i))$ for $i=0,\ldots,r-1$,
 ${\bf g}_0,\ldots,{\bf g}_{r-1}\notin\mathbb G^\ast$, and ${\bf g}_r\in\mathbb G^\ast$. Define a sequence of generalised fractional polymorphisms $\rho_0=\rho^\ast,\rho_1,\ldots,\rho_r$
as follows:
$$
\rho_{i+1}=\rho_i + \rho_i({\bf g}_i)\cdot(- \chi_{{\bf g}_i} +  \nu_i),\qquad \nu_i={\sf Exp}({\bf g}_i)
$$
for $i=0,\ldots,r-1$. We can prove by induction on $i$ that $\rho_i\in\Omega$, i.e.\ $\rho_i$ is a generalised fractional polymorphism of $\Gamma$.
Indeed, for any $f \in \Gamma$, and $x^1,\ldots,x^m\in \dom f$, we have
\begin{eqnarray*}
\sum_{{\bf h}\in\mathbb G} (-\chi_{{\bf g}_i}+\nu_i)({\bf h})f^{|\bf h|}({\bf h}(x^1,\ldots,x^m))
\le 0
\end{eqnarray*}
where we used eq.~\eqref{eq:valid} and the fact that $f^{|{\bf g}_i|}({\bf g}_i(x^1,\ldots,x^m))$ is finite
(since ${\bf g}_i\in\supp(\rho_i)$ and $\rho_i$ is a generalised fractional polymorphism of $f$). Therefore,

\begin{eqnarray*}
\sum_{{\bf h}\in\mathbb G} \rho_{i+1}({\bf h})f^{|\bf h|}({\bf h}(x^1,\ldots,x^m)) 
\le\sum_{{\bf h}\in\mathbb G} \rho_i({\bf h})f^{|\bf h|}({\bf h}(x^1,\ldots,x^m)) 
\le f^{m}({\bf g}(x^1,\ldots,x^m))
\end{eqnarray*}
(the last inequality is by the induction hypothesis). We proved that $\rho_i\in\Omega$ for all $i$.

By construction, $\rho^\ast(\mathbb G^\ast)=\rho_0(\mathbb G^\ast)\le \rho_1(\mathbb G^\ast)\le\ldots\le \rho_{r-1}(\mathbb G^\ast)<\rho_r(\mathbb G^\ast)$.
This contradicts the choice of $\rho^\ast$.
\end{proof}

%%%%%

\section{Second characterisation of general-valued languages}\label{sec:generate1}
%\section{Proofs of Theorem~\ref{th:BLP:generate} and Lemma~\ref{lemma:idem}} \label{sec:generate}

In this section, we use the Expansion Lemma from Section~\ref{sec:construction} to prove
Theorem~\ref{th:BLP:generate}. % and Lemma~\ref{lemma:idem}. 

\renewcommand{\thetheoremRESTATED}{\ref{th:BLP:generate}}
\begin{theoremRESTATED}[restated]
Suppose that, for every $n \geq 2$, 
$\Gamma$ admits a fractional polymorphism $\omega_n$
such that $\supp(\omega_n)$ generates a symmetric $n$-ary operation. 
Then, $\Gamma$ admits a symmetric fractional polymorphism of every arity $m \geq 2$.
\end{theoremRESTATED}

\begin{proof}
The proof is an application of Lemma~\ref{lemma:expansion}.
Fix some arbitrary arity $m \geq 2$.
Let $\sim$ denote the following equivalence relation on the
set $\calO^{(m)}$ of $m$-ary operations on $D$:
%For an $m$-ary operation $g$, let $\tilde{g}$ denote the equivalence class
%of $g$ under the relation:
\[
g \sim g' \Leftrightarrow \exists \pi \in S_m : g(x_1,\dots,x_m) = g'(x_{\pi(1)},\dots,x_{\pi(m)}).
%% \text{ for some $\pi \in S_m$.}
\]

Let $\mathbb G$ be the set of equivalence classes of the relation $\sim$ and
let ${\mathbb G}^\ast$ be the set of all equivalence classes ${\bf g} \in \mathbb G$
for which $|{\bf g}| = 1$, i.e., the set of equivalence classes containing a single symmetric
operation.
%% Note that we have $|\tilde{g}| = 1$ if and only if $g$ is symmetric.
%%
We say that a fractional operation $\nu$ is \emph{weight-symmetric} if
$\nu(g) = \nu(g')$ whenever $g \sim g'$.
A weight-symmetric fractional operation $\nu$ induces a probability distribution $\rho$ on $\mathbb G$:
$\rho({\bf g}) = \nu(g)/|{\bf g}|$, where $g$ is any of the operations in ${\bf g}$.

We now define the expansion operator ${\sf Exp}$ by giving its result when
applied to an arbitrary ${\bf g} \in \mathbb G$.
Let $n = |{\bf g}|$ and let $\omega$ be a $k$-ary fractional polymorphism of $\Gamma$
such that $\supp(\omega)$ generates a symmetric $n$-ary operation.

Define a sequence of $m$-ary weight-symmetric fractional operations
$\nu_0, \nu_1, \dots$, each with $\| \nu_i \|_1 = 1$, as follows.
Let $\nu_0 = \sum_{g \in {\bf g}} \frac{1}{|{\bf g}|} \chi_g$.
For $i \geq 1$, assume that $\nu_{i-1}$ has been defined.
Let $l_{i-1} = \min \{ \nu_{i-1}(g) \mid g \in \supp(\nu_{i-1}) \}$ be the
minimum weight of an operation in the support of $\nu_{i-1}$.
The fractional operation $\nu_i$ is obtained by subtracting from $\nu_{i-1}$
an equal weight from each operation in $\supp(\nu_{i-1})$ and
adding back this weight as superpositions of $\omega$ by all possible
choices of operations in $\supp(\nu_{i-1})$.
The amount subtracted from each operation is $\frac{1}{2} l_{i-1}$.
This implies that every collection in $\supp(\nu_{i-1})$ is also in $\supp(\nu_i)$.

Formally $\nu_i$ is defined as follows:
\[ %% \begin{equation}\label{eqn:nu-i}
\nu_i =
\nu_{i-1} -
\frac{l_{i-1}}{2} \chi_{i-1} +
\frac{l_{i-1}}{2} \eta_{i-1},
\] %% \end{equation}
where $\displaystyle\chi_{i-1} = \frac{1}{|\supp(\nu_{i-1})|} \sum_{g \in \supp(\nu_{i-1})} \chi_g$ and,
% is the (normalised) indicator function of the set
%$\supp(\nu_{i-1})$ and,
with $K = |\supp(\nu_{i-1})|^{k}$,
\[
\eta_{i-1} = \frac{1}{K} \sum_{g_1,\dots,g_k \in \text{supp}(\nu_{i-1})}
\omega[g_1,\dots,g_k].
\]

By definition, $\nu \geq 0$ and $\|\nu_i\|_1 = \|\nu_{i-1}\|_1 = 1$.
To see that $\nu_i$ is weight-symmetric, it suffices to verify that
$\eta_{i-1}$
is weight-symmetric.
Let $g$ be any $m$-ary operation of the form $g=h[g_1,\dots,g_k]$ and let
$g \sim g'$.
Let $\pi \in S_m$ be such that $g(x_1,\dots,x_m) =
g'(x_{\pi(1)},\dots,x_{\pi(m)})$ and define $g'_j(x_1,\dots,x_m) =
g_j(x_{\pi(1)},\dots,x_{\pi(m)})$ for $1 \leq j \leq k$.
Since $\nu_{i-1}$ is weight-symmetric, it follows that
$g_i \in \supp(\nu_{i-1})$ if and only if $g'_i \in \supp(\nu_{i-1})$.
Therefore the terms $\omega(h) h[g_1,\dots,g_k]$ in $\eta_{i-1}$ such that
$g = h[g_1,\dots,g_k]$ are in bijection with the terms
$\omega(h) h[g'_1,\dots,g'_k]$ such that $g' = h[g'_1,\dots,g'_k]$.
So the fractional operation $\eta_{i-1}$ assigns the same weight to $g$
and $g'$.

The assumption that $\supp(\omega)$ generates a symmetric $n$-ary operation $t$
means that $t$ can be obtained by a finite number of superpositions of operations
from $\supp(\omega)$ and the set of all projections.
Formally,
define ${\cal A}_0 = \{ e^{(i)}_j \mid 1 \leq j \leq i \}$ to be the set of all projections and,
for $j \geq 1$, define
${\cal A}_{j} = \{ g[h_1,\dots,h_q] \mid g \in \supp(\omega), h_1, \dots, h_q \in {\cal A}_{j'}, j' < j \}$.
Then, $t \in {\cal A}_d$, for some $d \geq 1$. 
Fix any such value of $d$ and define ${\sf Exp}({\bf g})$
to be the probability distribution on $\mathbb G$ induced by $\nu_d$.

We now show that this operator is non-vanishing.
Assume that ${\bf g} = \{g_1,\dots,g_n\}$.
Since $\supp(\nu_0) = \{g_1, \dots, g_n\}$ 
and using the fact that $\supp(\nu_i)$ contains all superpositions 
$g[h_1,\dots,h_q]$, $g \in \supp(\omega)$, $h_1, \dots, h_q \in \supp(\nu_{i-1})$,
it follows by induction that $t[g_1,\dots,g_n] \in \supp(\nu_d)$.
Note that the operation $t[g_1,\dots,g_n]$ is symmetric since
for all $\pi \in S_m$, there is a permutation $\pi' \in S_n$ such that
$
t[g_1,\dots,g_n](x_{\pi(1)},\dots,x_{\pi(m)}) = 
t[g_{\pi'(1)},\dots,g_{\pi'(n)}](x_1,\dots,x_m) =
t[g_1,\dots,g_n](x_1,\dots,x_m).
$
Hence, ${\sf Exp}({\bf g})$ assigns
non-zero probability to $\{t[g_1, \dots, g_n]\}$
and since $\{t[g_1, \dots, g_n]\} \in {\mathbb G}^\ast$,
it follows that ${\sf Exp}$ is non-vanishing.

It remains to show that ${\sf Exp}$ is valid for $\Gamma$, i.e., that (\ref{eq:valid}) is satisfied.
We claim that for each $i \geq 1$, we have
\begin{equation}\label{eq:claim}
\sum_{g \in \calO^{(m)}} \nu_{i}(g) f(g(x^1,\dots,x^m)) \leq \sum_{g \in \calO^{(m)}}
\nu_{i-1}(g) f(g(x^1,\dots,x^m)),
\end{equation}
for all $f \in \Gamma$ and
$x^1,\dots,x^m\in \dom f$.
The inequality (\ref{eq:valid}) then follows by induction on $i$,
noting that $\sum_{g \in \calO^{(m)}} \nu_0(g) f(g(x^1,\dots,x^m)) =
\sum_{g \in {\bf g}} \frac{1}{|{\bf g}|} f(g(x^1,\dots,x^m)) =
f^{|{\bf g}|}({\bf g}(x^1,\dots,x^m))$.
To see why (\ref{eq:claim}) holds, compare the last two terms in the definition of
$\nu_i$:
\begin{align*}
 \sum_{g \in \calO^{(m)}} \chi_{i-1}(g) f(g(x^1,\dots,x^m))
&= \frac{1}{K} \sum_{g_1,\dots,g_k \in \text{supp}(\nu_{i-1})} \frac{1}{k}
   \sum_{i=1}^k f(g_i(x^1,\dots,x^m)) \\
&\hspace{-5em}\geq \frac{1}{K} \sum_{g_1,\dots,g_k \in \text{supp}(\nu_{i-1})} \sum_{h \in \calO^{(m)}}  \omega(h) f(h[g_1,\dots,g_k](x^1,\dots,x^m)) \\
&\hspace{-5em}= \sum_{h' \in \calO^{(m)}} \eta_{i-1}(h') f(h'(x^1,\dots,x^m)).
\end{align*}

Hence, ${\sf Exp}$ is valid, so Lemma~\ref{lemma:expansion} is applicable and
shows that $\Gamma$ admits a generalised fractional polymorphism $\rho^\ast$ with
support on singleton sets, each containing a symmetric $m$-ary operation.
Therefore, $\Gamma$ admits the symmetric $m$-ary fractional polymorphism
$\sum_{\{g\} \in {\mathbb G}^\ast} \rho^\ast(\{g\}) \chi_g$.
\end{proof}

%%%%%%%%%%%%%%%%%%%%%%%%%%%%%%%%%%%%%%%%%%%%%%%%%%%%%%%%%%%%%%%%%%%%%%%%%%%%%%%%%%%%%%%%%%%%%%%%%

\section{Imposing idempotency}\label{sec:generate2}
%\section{Proofs of Theorem~\ref{th:BLP:generate} and Lemma~\ref{lemma:idem}} \label{sec:generate}

In this section we prove Lemma~\ref{lemma:idem} which was used
to find optimal solutions in Section~\ref{sec:selfreduce}.
The lemma states that for a symmetric fractional polymorphism, 
we can impose idempotency on a sub-domain $D'$ of $D$ while simultaneously
ensuring that there is an optimal solution with labels restricted to $D'$.
The proof uses the Expansion Lemma. 

\renewcommand{\thelemmaRESTATED}{\ref{lemma:idem}}
\begin{lemmaRESTATED}[restated]
There exists a subset $D' \subseteq D$ such that
if $\Gamma$ admits an $m$-ary symmetric fractional polymorphism, then it admits
an $m$-ary symmetric fractional polymorphism $\omega$ such that, for all $g \in
\supp(\omega)$,
\begin{enumerate}
\item\label{item:coreRESTATED}
$g(x,x,\dots,x) \in D'$ for all $x \in D$;
\item\label{item:idemRESTATED}
$g(x,x,\dots,x) = x$ for all $x \in D'$.
\end{enumerate}
\end{lemmaRESTATED}

\begin{proof}
%The proof is an application of Lemma~\ref{lemma:expansion}.
%% Again, we will apply Theorem~\ref{thm:tree-construction}.
%
Every language $\Gamma$ admits the fractional polymorphism that assigns
probability 1 to the unary identity operation on $D$.
Furthermore, assuming that $\Gamma$ admits two unary fractional polymorphisms
$\nu_1$ and $\nu_2$, $\Gamma$ also admits $\nu' = \frac{1}{2}(\nu_1+\nu_2)$
with  $\supp(\nu') = \supp(\nu_1) \cup \supp(\nu_2)$.
Therefore, we can let $\nu$ be a unary fractional polymorphism of $\Gamma$ with inclusion-maximal support.
Let $h \in \supp(\nu)$ be such that
$|h(D)| = \min \{ |g(D)| \mid g \in \supp(\nu) \}$ and define $D' = h(D)$.

Let $\mathbb G = \{ \{ g \} \mid g \in \calO^{(m)}_{\tt sym} \}$ and
let ${\mathbb G}^\ast$ be the operations in $\mathbb G$ that
additionally satisfy (\ref{item:coreRESTATED}) and (\ref{item:idemRESTATED}).
The expansion operatior ${\sf Exp}$ is defined as follows:
${\sf Exp}(\{g\})$ assigns probabilty
$\sum_{h' \in \supp(\nu), g' = h' \circ g} \nu(h')$ to the set $\{g'\}$.
It is easy to see that ${\sf Exp}$ is valid
and we show below that it is non-vanishing.
Therefore, Lemma~\ref{lemma:expansion} is applicable with
$\rho$ taken to be an $m$-ary symmetric fractional polymorphism of $\Gamma$.
Consequently, $\Gamma$ admits
$\omega = \sum_{\{g\} \in {\mathbb G}^\ast} \rho^\ast(\{g\}) \chi_g$.

We finish the proof by showing that ${\sf Exp}$ is non-vanishing.
It is easy to see that $\Gamma$ admits the unary fractional polymorphism $\mu = \sum_{h_1 \in \supp(\nu)}  \nu(h_1) \sum_{h_2 \in \supp(\nu)}\nu(h_2) \chi_{h_1 \circ h_2}$.
Since $\nu$ is inclusion-maximal, it follows that $h_1 \circ h_2 \in \supp(\mu) \subseteq \supp(\nu)$,
so $\supp(\nu)$ forms a monoid under composition.

Define $G$ as $\{ g|_{D'} \mid g \in \supp(\nu), g(D) = D' \}$.
Then, $G$ is a set of permutations of $D'$ that contains the identity.
Let $g'_1, g'_2 \in G$ be two permutations in this set
and let $g_1, g_2 \in \supp(\nu)$ be such that $g_1(D) = g_2(D) = D'$ and $g'_i = g_i|_{D'}$.
Since $\supp(\nu)$ forms a monoid under composition, $g_1 \circ g_2 \in \supp(\nu)$.
Therefore, $g'_1 \circ g'_2 = g_1|_{D'} \circ g_2|_{D'} = (g_1 \circ g_2)|_{D'} \in G$,
so $G$ forms a group under composition.

Let $\{g\} \in \mathbb G$.
By the \emph{diagonal} of an operation $f$ we mean the unary operation $x \mapsto f(x,\dots,x)$.
Note that the diagonal of $h \circ g$ acts as a permutation on $D'$.
This permutation has an inverse in $G$, so there exists an operation $i \in \supp(\nu)$ such that $i(D) = D'$ and such that the restriction of $i$ to $D'$ is the inverse of the diagonal of $h \circ g$.
Hence $\{i \circ h \circ g\} \in \mathbb G^\ast$.
Since $\supp(\nu)$ forms a monoid under composition,
$i \circ h \in \supp(\nu)$ so we conclude that $\{ i \circ h \circ g \} \in \supp({\sf Exp}(\{g\}))$.
\end{proof}

%%%%%%%%%%%%%%%%%%%%%%%%%%%%%%%%%%%%%%%%%%%%%%%%%%%%%%%%%%%%%%%%%%%%%%%%%%%%%%%%%%%%%%%%%%%%%%%%%
%%%%%%%%%%%%%%%%%%%%%%%%%%%%%%%%%%%%%%%%%%%%%%%%%%%%%%%%%%%%%%%%%%%%%%%%%%%%%%%%%%%%%%%%%%%%%%%%%
%%%%%%%%%%%%%%%%%%%%%%%%%%%%%%%%%%%%%%%%%%%%%%%%%%%%%%%%%%%%%%%%%%%%%%%%%%%%%%%%%%%%%%%%%%%%%%%%%

\section{Characterisation of finite-valued languages}
%\section{Proof of Theorem~\ref{thm:finite:gen}}
\label{sec:finite}

The goal of this section is to prove the characterisation of finite-valued languages
solved by BLP.
In particular, we prove the following theorem.

\renewcommand{\thetheoremRESTATED}{\ref{thm:finite:gen}}
\begin{theoremRESTATED}[restated]
Suppose that a finite-valued language $\Gamma$ admits
a symmetric fractional polymorphism of arity $m-1\ge 2$. Then
$\Gamma$ admits a symmetric fractional polymorphism of arity $m$.
\end{theoremRESTATED}

Let us fix a symmetric fractional polymorphism $\omega:\fpolMinusOne$ of $\Gamma$ of arity $m-1$.
We will use the letter $s$ for operations in $\supp(\omega)$ to emphasize that these operations are symmetric.

A symmetric fractional polymorphism of $\Gamma$ of arity $m$ will be constructed in two steps.
The first one will rely on the Expansion Lemma. Essentially, in this step we start
with a fractional polymorphism $\rho_0=\frac{1}{m}(e^{(m)}_1+\ldots+e^{(m)}_m)$
and repeatedly modify it by applying the fractional polymorphism $\omega$.
The example below demonstrates  such modification for $m=3$.

\begin{exampleX}\examplefont\label{ex:finiteValued}
Suppose that language $\Gamma$ admits a binary symmetric fractional polymorphism $\omega$.
For a function $f\in\Gamma$ and labellings $x,y,z\in\dom f$ we can write

{\small
\begin{eqnarray*}
f^3(x,y,z)
&\ =\ &\frac{1}{3}f^2(y,z) + \frac{1}{3}f^2(x,z) + \frac{1}{3}f^2(x,y) \nonumber \\
&\ge& \frac{1}{3}\sum_{s\in\supp(\omega)}\omega(s) f(s(y,z)) + \frac{1}{3}\sum_{s\in\supp(\omega)}\omega(s) f(s(x,z)) + \frac{1}{3}\sum_{s\in\supp(\omega)}\omega(s) f(s(x,y))\nonumber \\
&=& \sum_{s\in\supp(\omega)}\omega(s)f^3(s(y,z),s(x,z),s(x,y))
\end{eqnarray*}
}

This means that the following vector is a fractional polymorphism of $\Gamma$:
$$
\rho_1=\sum_{s\in\supp(\omega)}\omega(s)\cdot\frac{1}{3}(
\chi_{s\circ(e^{(3)}_2,e^{(3)}_3)} + 
\chi_{s\circ(e^{(3)}_1,e^{(3)}_3)} + 
\chi_{s\circ(e^{(3)}_1,e^{(3)}_2)}
) %,\qquad g_1=s\circ(e^{(3)}_2,e^{(3)}_3), g_2=s\circ(e^{(3)}_1,e^{(3)}_3), g_3=
$$
We can then take one component $\frac{1}{3}(\chi_{g_1}+\chi_{g_2}+\chi_{g_3})$ of the sum above
and replace it with another vector by applying $\omega$ in a similar way.
This shows how we can derive new fractional polymorphisms of $\Gamma$.
Note that such polymorphisms will have a special structure, namely they will be a weighted sum of vectors
of the form $\frac{1}{3}(\chi_{g_1}+\chi_{g_2}+\chi_{g_3})$ where $g_1,g_2,g_3\in\calO^{(3)}$.
This means that we will be working with the set $\mathbb G$ containing triplets of operations ${\bf g}=(g_1,g_2,g_3)\in\calO^{(3\rightarrow 3)}$.
Recall that in Section \ref{sec:construction} a probability distribution over such $\mathbb G$ was called a {\em fractional polymorphism of arity $3\rightarrow 3$}.
\end{exampleX}

The example above can be generalised to other values of $m\ge 3$ in a natural way.
The output of the first step (described in Section~\ref{sec:proof:finiteValued:1} below) will thus be a generalised fractional polymorphism of $\Gamma$
of arity $m\rightarrow m$ with certain properties that will be exploited in step 2.

In the second step (Sections~\ref{sec:proof:finiteValued:2}-\ref{sec:proof:finiteValued:5}) we will turn it into an $m$-ary symmetric fractional polymorphism of $\Gamma$
using tools such as Farkas' lemma. Note that in the second step the assumption that $\Gamma$ is finite-valued
will be essential.

With this introduction, we now proceed with the formal proof of Theorem~\ref{thm:finite:gen}.

\subsection{Proof of Theorem~\ref{thm:finite:gen}: Step 1}\label{sec:proof:finiteValued:1}

We start with some additional notation and definitions.
We use $[m]$ to denote set $\{1,\ldots,m\}$.
Let $\pi\in S_m$ be a permutation of $[m]$.
For a labelling $\alpha=(a_1,\ldots,a_m)\in D^m$ we define $\alpha^\pi\in D^m$ as follows: 
$\alpha^\pi=(a_{\pi(1)},\ldots,a_{\pi(m)})$.
For an operation $g:D^m\rightarrow D$, let $g^\pi:D^m\rightarrow D$ %and $g^h:D^m\rightarrow D$ 
be the following operation:
\begin{equation}
g^\pi(\alpha)=g(\alpha^\pi) %\qquad\quad g^h(\alpha)=g(\alpha^h)
\label{eq:gpi}
\end{equation}

For a symmetric operation  $s\in \supp(\omega)$ of arity $m-1$ we introduce the following definitions.
For a labelling $\alpha=(a_1,\ldots,a_m)\in D^m$ let  $\alpha^s\in D^m$
be the labelling
\begin{eqnarray}
\alpha^s&=&(s(\alpha_{-1}), \ldots, s(\alpha_{-m}))
\end{eqnarray}
where $\alpha_{-i}\in D^{m-1}$ is the labelling obtained from $\alpha$ by removing the $i$-th element.
For a mapping  ${\bf g}:D^m\rightarrow D^m$, let % (i.e.\ ${\bf g}=(g_1,\ldots,g_m)$ where $g_i:D^m\rightarrow D$), 
 ${\bf g}^s:D^m\rightarrow D^m$ be the mapping
\begin{subequations}
\begin{equation}
 {\bf g}^s(\alpha)=[{\bf g}(\alpha)]^s
\end{equation}
The last definition can also be expressed as
\begin{equation}
{\bf g}^s=(s\circ{\bf g}_{-1},\ldots,s\circ{\bf g}_{-m})
\label{eq:gh}
\end{equation}
\end{subequations}
where ${\bf g}_{-i}:D^m\rightarrow D^{m-1}$ is the sequence of $m-1$ operations obtained from ${\bf g}=(g_1,\ldots,g_m)$ by removing the $i$-th operation. We use ${\bf g}^{s_1\ldots s_k}$ to denote $(\ldots({\bf g}^{s_1})^{\ldots})^{s_k}$.

Let $\mathds{1}$ be the identity mapping $D^m\rightarrow D^m$, and let
$\mathbb G=\{{\mathds 1}^{s_1\ldots s_k}\:|\:s_1,\ldots,s_k\in \supp(\omega),k\ge 0\}\subseteq \calO^{(m\rightarrow m)}$ be the set of all mappings that can be obtained from $\mathds{1}$ by applying operations from $\supp(\omega)$.

\paragraph{Graph on mappings} Let us define a directed weighted graph $(\mathbb G,E,w)$ with the set of edges 
$E=\{({\bf g},{\bf g}^s)\:|\:{\bf g}\in \mathbb G,s\in \supp(\omega)\}$ and positive weights
$w({\bf g},{\bf h})=\sum_{s\in \supp(\omega):{\bf h}={\bf g}^s}\omega(s)$ for $({\bf g},{\bf h})\in E$.
Clearly, we have 
\begin{equation}
\sum_{{\bf h}:({\bf g},{\bf h})\in E}w({\bf g},{\bf h})=1\qquad\quad\forall{\bf g}\in \mathbb G
\label{eq:weightsSumToOne}
\end{equation}
The graph $({\mathbb G,E})$ can be decomposed into strongly connected components, yielding  a directed acyclic
graph (DAG) on these components. 
We define ${\sf Sinks}({\mathbb G,E})$ to be the set of those strongly connected components $\mathbb H\subseteq \mathbb G$
of $(\mathbb G,E)$ that are sinks of this DAG (i.e.\ have no outgoing edges). Any DAG has at least one sink, therefore ${\sf Sinks}({\mathbb G,E})$ is non-empty.
%, i.e.\ all edges in $(\mathbb G,E)$ from $\mathbb H$ lead to vertices in $\mathbb H$.
We denote $\mathbb G^\ast=\bigcup_{\mathbb H\in\mathbb {\sf Sinks}({\mathbb G,E})} \mathbb H\subseteq \mathbb G$.

By applying the Expansion Lemma to the sets of collections $(\mathbb G,\mathbb G^\ast)$ defined above
we can obtain the following result.
\begin{lemmaX}
%Denote $\widehat H=\bigcup_{\calH\in\mathbb H[G]}\subseteq V$.
% be the union of all strongly connected components of $G$.
There exists a generalised fractional polymorphism $\rho^\ast$ of $\Gamma$ of arity $m\rightarrow m$
with $\supp(\rho^\ast)\subseteq \mathbb G^\ast$.
\label{th:rho:suppH}
\end{lemmaX}
\begin{proof}
%To prove the existence of a generalised fractional polymorphism $\rho^\ast$ of $\Gamma$ with $\supp(\rho^\ast)\in\mathbb G^\ast$, we will use Lemma~\ref{lemma:expansion}. 
Clearly, $\Gamma$ admits a least one
generalised
fractional polymorphism $\rho$ with $\supp(\rho)\subseteq\mathbb G$, namely $\rho=\chi_{\mathds{1}}$.
It thus suffices to prove the existence of an expansion operator {\sf Exp} which is valid for $\Gamma$
and non-vanishing with respect to $(\mathbb G,\mathbb G^\ast)$.

Given a mapping ${\bf g}\in\mathbb G$, we define the probability distribution $\rho={\sf Exp}({\bf g})$ as follows:
$$
\rho=\sum_{s\in\supp(\omega)} \omega(s) \chi_{{\bf g}^s}
$$
Let us check that it is indeed valid for $\Gamma$. Consider a function $f\in\Gamma$ of arity $n$ and labellings $x^1,\ldots,x^m\in D^n$.
Denote $(y^1,\ldots,y^m)={\bf g}(x^1,\ldots,x^m)$.
Then,
\begin{eqnarray*}
\sum_{{\bf h}\in \supp(\rho)} \rho({\bf h}) f^m({\bf h}(x^1,\ldots,x^m))
&=&\sum_{s\in \supp(\omega)}\omega(s)f^m({\bf g}^s(x^1,\ldots,x^m)) \\
&=&\sum_{s\in \supp(\omega)}\omega(s)\frac{1}{m}\sum_{i\in[m]}f(s((y^1,\ldots,y^m)_{-i})) \\
&=&\frac{1}{m}\sum_{i\in[m]}\sum_{s\in \supp(\omega)}\omega(s)f(s((y^1,\ldots,y^m)_{-i})) \\
&\le& \frac{1}{m}\sum_{i\in[m]}f^{m-1}((y^1,\ldots,y^m)_{-i}))\\
&=& f^{m}(y^1,\ldots,y^m)
\;\;=\;\; f^{m}({\bf g}(x^1,\ldots,x^m)).
\end{eqnarray*}

Now let us show that the expansion operator {\sf Exp} is non-vanishing.
Observe that $\supp({\sf Exp}({\bf g}))=\{{\bf h}\:|\:({\bf g},{\bf h})\in E\}$
for any ${\bf g}\in\mathbb G$. Furthermore, it follows from the definition of $\mathbb G^\ast$
that for any ${\bf g}\in\mathbb G$ there exists a path in $(\mathbb G,E)$ from
${\bf g}$ to some node ${\bf g}^\ast\in\mathbb G^\ast$. These two facts imply
the claim. 
\end{proof}

This concludes the first step of the proof. To summarize, we have constructed a generalised
fractional polymorphism $\rho^\ast$ of $\Gamma$ with $\supp(\rho^\ast)\subseteq\mathbb G^\ast$.
Note that operations in collections ${\bf g}\in\supp(\rho^\ast)$ are not necessarily
symmetric (otherwise this would be a contradiction to Example~\ref{example:STPcycle}).
In the second step we will show that for finite-valued languages we can replace these
collections with ${\bf p}\circ{\bf g}$,
where ${\bf p}:D^m\rightarrow D^m$ is a mapping that 
orders tuples $\alpha=(a_1,\ldots,a_m)\in D^m$ according to some total order on $D$.
More precisely, we will show that 
$$
f^m({\bf g}(x^1,\ldots,x^m))=f^m(({\bf p}\circ{\bf g})(x^1,\ldots,x^m))
\qquad\forall{\bf g}\in\mathbb G^\ast,f\in\Gamma,x^1,\ldots,x^m\in\dom f
$$
This will imply that vector 
$\rho = \sum_{{\bf g}\in \supp(\rho^\ast)} 
\rho^\ast({\bf g})\chi_{{\bf p}\circ{\bf g}}$
is also a generalised fractional polymorphism of $\Gamma$, which gives
an $m$-ary symmetric fractional polymorphism of $\Gamma$, thus proving Theorem~\ref{thm:finite:gen}.

\subsection{Proof of Theorem~\ref{thm:finite:gen}: Step 2}\label{sec:proof:finiteValued:2}
We start with the following observation.
\begin{propositionX}
\label{prop:cluster:basic} %(a) There holds $g^{h_1\ldots h_k}(\alpha)=g(\alpha^{h_k\ldots h_1})$ for $g\in\calH$, $h_1,\ldots,h_k\in \supp(\omega)$. \\
%(a) There holds $\mathds{1}^{h_1\ldots h_k}\circ {\bf g}={\bf g}^{h_1\ldots h_k}$
%for all $h_1,\ldots,h_k\in \supp(\omega)$ where $\circ$ denotes the composition of mappings $D^m\rightarrow D^m$. Consequently,
%$\calH$ is closed under composition, and so it is a semigroup with identity. There holds $\langle{\bf g}\rangle=\{{\bf g}'\circ{\bf g}\:|\:{\bf g}'\in \calH\}$ for ${\bf g}\in\calH$. \\
%(b) 
Every ${\bf g}=(g_1,\ldots,g_m)\in \mathbb G$ satisfies the following:
\begin{eqnarray}
(g^\pi_1,\ldots,g^\pi_m)&=&(g_{\pi(1)},\ldots,g_{\pi(m)}) \quad \qquad\quad\forall \mbox{ permutations $\pi\in S_m$} \label{eq:mcluster:a} 
\end{eqnarray}
%\end{subequations}
%(c) If ${\bf g}\in\calH$, $h\in \supp(\omega)$ and $\alpha\in\Delta^{(m)}$ then ${{\bf g}^h}(\alpha)=[{\bf g}(\alpha)]^h$.
\end{propositionX}

\noindent Thus, permuting the arguments of $g_i(\cdot,\ldots,\cdot)$ gives a mapping which is also present in the sequence ${\bf g}$,
possibly at a different position.
\begin{proof} 
%(a) For each $\alpha\in D^m$ we have
%$$
%[ \mathds{1}^{h_1\ldots h_{k}}\circ{\bf g}](\alpha)
%= \mathds{1}^{h_1\ldots h_{k}}({\bf g}(\alpha))
%=  [\mathds{1}({\bf g}(\alpha))]^{h_1\ldots h_{k}}
%=  [{\bf g}(\alpha)]^{h_1\ldots h_{k}}
%=  {\bf g}^{h_1\ldots h_{k}}(\alpha)
%$$
%\\
%(b) 
Checking that $\mathds{1}$ satisfies \eqref{eq:mcluster:a} is straightforward.
Let us prove that for any ${\bf g}:D^m\rightarrow D^m$ satisfying \eqref{eq:mcluster:a} and for any symmetric
operation $s\in\calO^{(m-1)}$, the mapping ${\bf g}^s$ also
satisfies \eqref{eq:mcluster:a}. 
Consider $i\in[m]$. We need to show that $(s\circ{\bf g}_{-i})^\pi=s\circ{\bf g}_{-\pi(i)}$.
For each $\alpha\in D^m$ we have
%Consider $\alpha=(a_1,\ldots,a_m)$, and denote $\beta=(a_{\pi(1)},\ldots,a_{\pi(m)})$. We can write
\begin{eqnarray*}
(s\circ{\bf g}_{-i})^\pi(\alpha)
&=&s\circ{\bf g}_{-i}(\alpha^\pi) \\
&=&s(g_1(\alpha^\pi),\ldots,g_{i-1}(\alpha^\pi),g_{i+1}(\alpha^\pi),\ldots,g_m(\alpha^\pi)) \\
&=& s(g^\pi_1(\alpha),\ldots,g^\pi_{i-1}(\alpha),g^\pi_{i+1}(\alpha),\ldots,g^\pi_m(\alpha)) \\
&=& s(g_{\pi(1)}(\alpha),\ldots,g_{\pi(i-1)}(\alpha),g_{\pi(i+1)}(\alpha),\ldots,g_{\pi(m)}(\alpha)) = s\circ{\bf g}_{-\pi(i)}(\alpha)
\end{eqnarray*}
%
%g^\pi_i(\alpha)
\end{proof}

For the next statement consider a connected component $\mathbb H\in{\sf Sinks}({\mathbb G,E})$,
and denote $I=\mathbb H\times [m]$. Below we use the Iverson bracket notation: $[\phi]=1$  if $\phi$ is true,
and $[\phi]=0$ otherwise.

\begin{lemmaX}\label{lemma:lambdaWeights}
(a) For any fixed distinct ${\bf g}',{\bf g}''\in\mathbb H$, there exists a vector $\lambda\in\mathbb R^{I}_{\ge 0}$ that satisfies %~\eqref{eq:lambdaWeights}
\begin{eqnarray}
\sum_{{\bf h}:({\bf h},{\bf g})\in E}w({\bf h},{\bf g})\lambda_{{\bf h} i}
-\sum_{j\in[m]-\{i\}}\frac{\lambda_{{\bf g} j}}{m-1} &=& c_{{\bf g}} 
\hspace{20pt} \forall ({\bf g},i)\in I 
\label{eq:lambdaWeights:b}
\end{eqnarray}
where $c_{\bf g}=[{\bf g}={\bf g}']-[{\bf g}={\bf g}'']$. \\
%and $[\cdot]$ is the Iverson bracket: it equals $1$ if the argument is true, and $0$ otherwise. \\
(b) For any fixed distinct $i',i''\in[m]$, there exists a vector $\lambda\in\mathbb R^{I\cup \mathbb H}_{\ge 0}$ that satisfies
\begin{subequations}
\begin{eqnarray}
\sum_{{\bf h}:({\bf h},{\bf g})\in E}w({\bf h},{\bf g})\lambda_{{\bf h} i}
-\sum_{j\in[m]-\{i\}}\frac{\lambda_{{\bf g} j}}{m-1} &=& c_i \lambda_{{\bf g}}
\hspace{20pt} \forall ({\bf g},i)\in I \label{eq:lambdaWeights} \\
\sum_{{\bf g} \in \mathbb H}\lambda_{{\bf g}} &=& 1 
\end{eqnarray}
\end{subequations}
where $c_i=[i=i']-[i=i'']$.
\end{lemmaX}

A proof of this statement is given in Section~\ref{sec:proof:finiteValued:3}, and is based on Farkas' Lemma.

Now let us fix a function $f\in\Gamma$ of arity $n$.
% and a connected component $\mathbb H\in{\sf Sinks}({\mathbb G,E})$. We denote $I=\mathbb H\times[m]$. %$I=\{({\bf g},i)\:|\:{\bf g}\in H,i\in[m]\}$.
Given labellings $x^1,\ldots,x^m$, we define labellings $x^{{\bf g}i}$ for all $({\bf g},i)\in I$ via
\begin{equation}
(x^{{\bf g} 1},\ldots,x^{{\bf g} m})={\bf g}(x^1,\ldots,x^m)%\qquad\quad\forall v\in[n]
\label{eq:xgi:def}
\end{equation}
Note that $x^{{\bf g}i}$ is a function of $(x^1,\ldots,x^m)$; for brevity of notation, this dependence is not shown.
For a vector $\lambda\in\mathbb R^{\mathbb H}$ and an index $i\in[m]$ we define the function $F^\lambda_i$ via
\begin{equation}
F^\lambda_i(x^1,\ldots,x^m)=\sum_{{\bf g}\in \mathbb H}\lambda_{\bf g}f(x^{{\bf g}i})
\qquad\quad
\forall x^1,\ldots,x^m\in D^n
%(x^{{\bf g} 1},\ldots,x^{{\bf g} m})={\bf g}(x^1,\ldots,x^m)%\qquad\quad\forall v\in[n]
\label{eq:Fi:def}
\end{equation}

\begin{lemmaX} % Consider $\mathbb H\in\mathbb {\sf Sinks}({\mathbb G,E})$. \\ % and $x^1,\ldots,x^m\in D^n$.\\
(a) It holds that $f^m(x^{{\bf g}'1},\ldots,x^{{\bf g}'m})=f^m(x^{{\bf g}''1},\ldots,x^{{\bf g}''m})$ for all
 ${\bf g}',{\bf g}''\in\mathbb H$ and $x^1,\ldots,x^m\in D^n$. % (Recall that $f^m(y^1,\ldots,y^m)=\frac{1}{m}\sum_{i\in[m]}f(y^i)$.)
\\
(b) There exists a probability distribution $\lambda$ over $\mathbb H$ 
such that
$F^\lambda_{i'}(x^1,\ldots,x^m)=F^\lambda_{i''}(x^1,\ldots,x^m)$ for all $i',i''\in[m]$ and $x^1,\ldots,x^m\in D^n$.
\label{th:cyclicPol}
\end{lemmaX}

A proof of Lemma~\ref{th:cyclicPol} is given in section~\ref{sec:proof:finiteValued:4}. 
The idea of the proof is as follows. 
%First, we will replace equalities  in the lemma with inequalities; since these inequalities must hold for all indexes, this change will not affect the claim.
Let us fix ${\bf g}',{\bf g}''$ and labellings $x^1,\ldots,x^m\in D^n$.
We will write down inequalities for the fractional polymorphism $\omega$ applied
to $m-1$ labellings $(x^{{\bf g} 1},\ldots,x^{{\bf g} m})_{-i}$ with $({\bf g},i)\in I$.
We will then take a linear combination of these inequalities with weights $\lambda_{{\bf g}i}\ge 0$
constructed in Lemma~\ref{lemma:lambdaWeights}(a); this will give inequality $f^m(x^{{\bf g}'1},\ldots,x^{{\bf g}'m})\le f^m(x^{{\bf g}''1},\ldots,x^{{\bf g}''m})$.
This inequality should hold for all choices of  ${\bf g}',{\bf g}''$, therefore it must actually be an equality.
Part (b) of Lemma \ref{th:cyclicPol} will be proved in a similar way.

With Lemma~\ref{th:cyclicPol} we will finally be able to prove the following (see Section~\ref{sec:proof:finiteValued:5}).
\begin{lemmaX}
Let  ${\bf g}^\ast$ be a mapping in $\mathbb G^\ast$ and
 ${\bf p}\in\calO^{(m\rightarrow m)}$ be any mapping such
that   ${\bf p}(\alpha)$ is a permutation of $\alpha$ for all $\alpha\in D^m$.
Denote
$$Range_n({\bf g}^\ast)=\{{\bf g}^\ast(x^1,\ldots,x^m)\:|\:x^1,\ldots,x^m\in D^n\}$$
For any function $f\in\Gamma$ of arity $n$ and any $(x^1,\ldots,x^m)\in Range_n({\bf g}^\ast)$ it holds that $f^m(x^1,\ldots,x^m)=f^m({\bf p}(x^1,\ldots,x^m))$.
\label{th:permutationInvariance}
\end{lemmaX}

This will imply Theorem~\ref{thm:finite:gen}.
Indeed, we can construct an $m$-ary symmetric fractional polymorphism of $\Gamma$
as follows. Take the vector $\rho^\ast$ from Lemma~\ref{th:rho:suppH}, take a  mapping ${\bf p}\in\calO^{(m\rightarrow m)}$
that orders tuples $\alpha=(a_1,\ldots,a_m)\in D^m$ according to some total order on $D$,
%satisfying the condition of Lemma~\ref{th:permutationInvariance}, 
and define the following vector.
$$
\rho=\sum_{{\bf g}\in \supp(\rho^\ast)}\rho^\ast({\bf g})\chi_{{\bf p}\circ{\bf g}}
$$
Then, $\Gamma$ admits $\rho$ since for any $f \in \Gamma$ and
for any  labellings $x^1,\ldots,x^m\in D^n$ we have
%\small
\begin{eqnarray*}
\sum_{{\bf h}\in \supp(\rho)} \!\!\!\!\rho({\bf h})f^m({\bf h}(x^1,\ldots,x^m))
&=&\!\!\!\sum_{{\bf g}\in \supp(\rho^\ast)}\!\!\! \rho^\ast({\bf g})f^m({\bf p}({\bf g}(x^1,\ldots,x^m)))\\
&=&\!\!\!\sum_{{\bf g}\in \supp(\rho^\ast)}\!\!\! \rho^\ast({\bf g})f^m({\bf g}(x^1,\ldots,x^m))
\le f^m(x^1,\ldots,x^m)
\end{eqnarray*}
\normalsize
Note, for any ${\bf h}=(h_1,\ldots,h_m)\in \supp(\rho)$, the operations $h_1,\ldots,h_m$ are symmetric.
Indeed, we have ${\bf h}={\bf p}\circ{\bf g}$ for some ${\bf g}\in \mathbb G^\ast$.
If $\alpha\in D^m$ and $\pi$ is a permutation of the set $[m]$ then  
${\bf h}(\alpha^\pi)={\bf p}({\bf g}(\alpha^\pi))={\bf p}([{\bf g}(\alpha)]^\pi)={\bf p}({\bf g}(\alpha))={\bf h}(\alpha)$
which implies the claim.

A symmetric $m$-ary fractional polymorphism of $\Gamma$ is finally given by
$\omega_m=\sum_{{\bf g}\in{\mathbb G}} \rho({\bf g}) \sum_{g\in{\bf g}} \frac{1}{m} \chi_g$.

It remains to prove Lemmas~\ref{lemma:lambdaWeights}, \ref{th:cyclicPol} and \ref{th:permutationInvariance}.

%%%%%%%%%%%%%%%%%%%%%%%%%%%%%%%%%%%%%%%%%%%%%%%%%%%%%%%%%%%%%%%%%%%%%%%%%%%%%%%%%%%%%%%%
%%%%%%%%%%%%%%%%%%%%%%%%%%%%%%%%%%%%%%%%%%%%%%%%%%%%%%%%%%%%%%%%%%%%%%%%%%%%%%%%%%%%%%%%
%%%%%%%%%%%%%%%%%%%%%%%%%%%%%%%%%%%%%%%%%%%%%%%%%%%%%%%%%%%%%%%%%%%%%%%%%%%%%%%%%%%%%%%%

\subsection{Proof of Lemma~\ref{lemma:lambdaWeights}}\label{sec:proof:finiteValued:3}

\renewcommand{\thelemmaRESTATED}{\ref{lemma:lambdaWeights}}
The lemma has two parts. For each part we will use a similar technique, namely we will assume
the opposite and then derive a contradiction by using Farkas' lemma.
\begin{lemmaRESTATED}[restated] %\label{lemma:lambdaWeights}
(a) There exists a vector $\lambda\in\mathbb R^{I}_{\ge 0}$ that satisfies %~\eqref{eq:lambdaWeights}
\begin{eqnarray}
\sum_{{\bf h}:({\bf h},{\bf g})\in E}w({\bf h},{\bf g})\lambda_{{\bf h} i}
-\sum_{j\in[m]-\{i\}}\frac{\lambda_{{\bf g} j}}{m-1} &=& c_{{\bf g}} 
\hspace{20pt} \forall ({\bf g},i)\in I 
%\label{eq:lambdaWeights:b}
\end{eqnarray}
where $c_{\bf g}=[{\bf g}={\bf g}']-[{\bf g}={\bf g}'']$.
%and $[\cdot]$ is the Iverson bracket: it equals $1$ if the argument is true, and $0$ otherwise. \\
%(b) There exists a vector $\lambda\in\mathbb R^{I\cup \mathbb H}_{\ge 0}$ that satisfies
%\begin{subequations}
%\begin{eqnarray}
%\sum_{{\bf h}:({\bf h},{\bf g})\in E}w({\bf h},{\bf g})\lambda_{{\bf h} i}
%-\sum_{j\in[m]-\{i\}}\frac{\lambda_{{\bf g} j}}{m-1} &=& c_i \lambda_{{\bf g}}
%\hspace{20pt} \forall ({\bf g},i)\in I \label{eq:lambdaWeights} \\
%\sum_{{\bf g} \in \mathbb H}\lambda_{{\bf g}} &=& 1 
%\end{eqnarray}
%\end{subequations}
%where $c_i=[i=i']-[i=i'']$.
\end{lemmaRESTATED}
\begin{proof}
%{\bf Part {(a)}}~
Suppose that the claim does not hold. By Farkas' lemma~\cite{Schrijver86:ILP} there exists 
a vector $y\in{\mathbb R}^I$ %and scalar $z\in\mathbb R$
such that
\begin{subequations}\label{eq:partb:Farkas}
\begin{eqnarray}
\sum_{({\bf g},i)\in I} c_{{\bf g}}y_{{\bf g}i} &<& 0 \label{eq:partb:Farkas:a} \\
\sum_{{\bf h}:({\bf g},{\bf h})\in E}w({\bf g},{\bf h})y_{{\bf h} i} 
-\sum_{j\in[m]-\{i\}}\frac{y_{{\bf g} j}}{m-1} &\ge & 0\hspace{20pt} \forall ({\bf g},i)\in I\label{eq:parta:LAJHGASFAFA}
\end{eqnarray}
\end{subequations}
Denote $u_{\bf g}=\sum_{i\in[m]}y_{{\bf g} i}$. 
Summing inequalities \eqref{eq:parta:LAJHGASFAFA} over $i\in[m]$ gives
\begin{eqnarray}
\sum_{{\bf h}:({\bf g},{\bf h})\in E}w({\bf g},{\bf h})u_{{\bf h}} 
-u_{\bf g} &\ge & 0\hspace{20pt} \forall {\bf g}\in \mathbb H\label{eq:GIKJHDF}
\end{eqnarray}
Denote ${\mathbb H}^\ast=\arg\max \{u_{\bf g}\:|\:{\bf g}\in \mathbb H\}$. From~\eqref{eq:weightsSumToOne} and~\eqref{eq:GIKJHDF}
we conclude that ${\bf g}\in {\mathbb H}^\ast$ implies ${\bf h}\in {\mathbb H}^\ast$ for all $({\bf g},{\bf h})\in E$. Therefore, ${\mathbb H}^\ast=\mathbb H$ (since $\mathbb H$ is a strongly connected component of $G$).

We showed that $u_{\bf g}=C$ for all ${\bf g}\in \mathbb H$ where $C\in\mathbb R$ is some constant.
But then the expression on the LHS of~\eqref{eq:partb:Farkas:a} equals $C-C=0$, a contradiction.
\end{proof}

\begin{lemmaRESTATED}[restated] %\label{lemma:lambdaWeights}
%(a) There exists a vector $\lambda\in\mathbb R^{I}_{\ge 0}$ that satisfies %~\eqref{eq:lambdaWeights}
%\begin{eqnarray}
%\sum_{{\bf h}:({\bf h},{\bf g})\in E}w({\bf h},{\bf g})\lambda_{{\bf h} i}
%-\sum_{j\in[m]-\{i\}}\frac{\lambda_{{\bf g} j}}{m-1} &=& c_{{\bf g}} 
%\hspace{20pt} \forall ({\bf g},i)\in I 
%\label{eq:lambdaWeights:b}
%\end{eqnarray}
%where $c_{\bf g}=[{\bf g}={\bf g}']-[{\bf g}={\bf g}'']$
%and $[\cdot]$ is the Iverson bracket: it equals $1$ if the argument is true, and $0$ otherwise. \\
(b) There exists a vector $\lambda\in\mathbb R^{I\cup \mathbb H}_{\ge 0}$ that satisfies
\begin{subequations}
\begin{eqnarray}
\sum_{{\bf h}:({\bf h},{\bf g})\in E}w({\bf h},{\bf g})\lambda_{{\bf h} i}
-\sum_{j\in[m]-\{i\}}\frac{\lambda_{{\bf g} j}}{m-1} &=& c_i \lambda_{{\bf g}}
\hspace{20pt} \forall ({\bf g},i)\in I \\ %\label{eq:lambdaWeights} \\
\sum_{{\bf g} \in \mathbb H}\lambda_{{\bf g}} &=& 1 
\end{eqnarray}
\end{subequations}
where $c_i=[i=i']-[i=i'']$.
\end{lemmaRESTATED}

\begin{proof}
Suppose that the claim does not hold. By Farkas' lemma~\cite{Schrijver86:ILP} there exist 
a vector $y\in{\mathbb R}^I$ and a scalar $z\in\mathbb R$
such that
\begin{subequations}\label{eq:Farkas}
\begin{eqnarray}
z &<& 0 \label{eq:Farkas:a} \\
z - \sum_{i\in[m]}c_i y_{{\bf g} i}&\ge& 0 \hspace{20pt} \forall {\bf g}\in \mathbb H \label{eq:Farkas:b}\\
\sum_{{\bf h}:({\bf g},{\bf h})\in E}w({\bf g},{\bf h})y_{{\bf h} i} 
-\sum_{j\in[m]-\{i\}}\frac{y_{{\bf g} j}}{m-1} &\ge & 0\hspace{20pt} \forall ({\bf g},i)\in I\label{eq:LAJHGASFAFA}
\end{eqnarray}
\end{subequations}
Denote $u_{\bf g}=\sum_{i\in[m]}y_{{\bf g} i}$. 
Using the same argument as in part (a) 
we conclude that $u_{\bf g}=C$ for all ${\bf g}\in \mathbb H$ where $C\in\mathbb R$ is some constant.
%We showed that $u_{\bf g}=C$ for all ${\bf g}\in \calH$ where $C\in\mathbb R$ is some constant.
We can assume w.l.o.g.\ that this constant is zero. Indeed, this can be achieved by subtracting $C/m$
from values  $y_{{\bf g} i}$ for all $({\bf g},i)\in I$ with ${\bf g}\in\mathbb H$; it can be checked (using eq.~\eqref{eq:weightsSumToOne})
that this operation preserves inequalities~\eqref{eq:Farkas}.
We thus have
\begin{eqnarray}
\sum_{j\in[m]-\{i\}}y_{{\bf g} j}&=&-y_{{\bf g} i}\qquad\quad\forall ({\bf g},i)\in I
\label{eq:zSumsToOne}
\end{eqnarray}
Substituting this into~\eqref{eq:LAJHGASFAFA} gives
\begin{eqnarray}
\sum_{{\bf h}:({\bf g},{\bf h})\in E}w({\bf g},{\bf h})y_{{\bf h} i} 
+\frac{y_{{\bf g} i}}{m-1} &\ge & 0\hspace{20pt} \forall ({\bf g},i)\in I\label{eq:GJKSADF}
\end{eqnarray}
Consider $k\in[m]$. Summing~\eqref{eq:GJKSADF} over $i\in[m]-\{k\}$ and then using~\eqref{eq:zSumsToOne} yields
\begin{eqnarray}
\sum_{{\bf h}:({\bf g},{\bf h})\in E}w({\bf g},{\bf h})(-y_{{\bf h} k})
+\frac{-y_{{\bf g} k}}{m-1} &\ge & 0\hspace{20pt} \forall ({\bf g},k)\in I\label{eq:GJKSADF'}
\end{eqnarray}
Combining~\eqref{eq:GJKSADF} and \eqref{eq:GJKSADF'} gives
\begin{eqnarray}
\sum_{{\bf h}:({\bf g},{\bf h})\in E}w({\bf g},{\bf h})y_{{\bf h} i} 
+\frac{y_{{\bf g} i}}{m-1} &= & 0\hspace{20pt} \forall ({\bf g},i)\in I\label{eq:GJKSADF''}
\end{eqnarray}
Denote $r_{{\bf g}}=\sum_{i\in[m]}c_i y_{{\bf g} i}$ for ${\bf g}\in \mathbb H$. Summing \eqref{eq:GJKSADF''}
over $i\in[m]$ with appropriate coefficients gives
\begin{eqnarray}
\sum_{{\bf h}:({\bf g},{\bf h})\in E}w({\bf g},{\bf h})r_{{\bf h}} 
+\frac{r_{{\bf g}}}{m-1} &= & 0\hspace{20pt} \forall {\bf g}\in \mathbb H\label{eq:GHAD}
\end{eqnarray}
From \eqref{eq:Farkas:a} and \eqref{eq:Farkas:b} we conclude that $r_{\bf g}<0$ for all ${\bf g}\in \mathbb H$, and thus eq.~\eqref{eq:GHAD}
cannot hold, a contradiction.
\end{proof}

%%%%%%%%%%%%%%%%%%%%%%%%%%%%%%%%%%%%%%%%%%%%%%%%%%%%%%%
\subsection{Proof of Lemma~\ref{th:cyclicPol}}\label{sec:proof:finiteValued:4}
Let us fix  a function $f\in\Gamma$ of arity $n$ and a connected component $\mathbb H\in{\sf Sinks}({\mathbb G,E})$.
In this subsection we will prove the following. 
\renewcommand{\thelemmaRESTATED}{\ref{th:cyclicPol}}
\begin{lemmaRESTATED}[equivalent statement]
(a) Inequality
\begin{subequations}
\begin{eqnarray}
\sum_{i\in[m]}f(x^{{\bf g}'i})-\sum_{i\in[m]}f(x^{{\bf g}''i})\le 0
\label{eq:cyclicPol:a}
\end{eqnarray}
holds for any distinct mappings ${\bf g}',{\bf g}''\in \mathbb H$  and any $x^1,\ldots,x^m\in D^n$.
\\
(b) There exists a probability distribution $\lambda$ over $\mathbb H$ 
such that 
\begin{eqnarray}
\sum_{{\bf g}\in \mathbb H}\lambda_{\bf g}f(x^{{\bf g}i'})-\sum_{{\bf g}\in \mathbb H}\lambda_{\bf g}f(x^{{\bf g}i''})\le 0
\label{eq:cyclicPol:b}
\end{eqnarray}
\end{subequations}
 for any distinct indices $i',i''\in[m]$ and any $x^1,\ldots,x^m\in D^n$.
\end{lemmaRESTATED}

Since inequality \eqref{eq:cyclicPol:a} holds for any pair of distinct mappings ${\bf g}',{\bf g}''\in \mathbb H$,
we conclude that in \eqref{eq:cyclicPol:a} we actually must have an equality (and similarly for \eqref{eq:cyclicPol:b}).
Therefore, the statement above is indeed equivalent to the original formulation of Lemma~\ref{th:cyclicPol},
which had equalities. (Note, we have also moved terms from the right-hand side of the original equalities to the left-hand side
with the negative sign; for that we have used the fact the $\Gamma$ is finite-valued.)

We will need the following observation.
\begin{propositionX}
If $\bh=\bg^s$ where $\bg\in \mathbb H$, $s\in \supp(\omega)$ then $x^{{\bf h} i}=s((x^{{\bf g} 1},\ldots,x^{{\bf g} m})_{-i})$ for $i\in[m]$
where $(x^{{\bf g} 1},\ldots,x^{{\bf g} m})_{-i}$ is the sequence of $m-1$ labellings obtained by removing the $i$-th labelling.
\label{prop:asjgasdg}
\end{propositionX}
\begin{proof} Consider a coordinate $v\in [n]$, and denote $\alpha=(x^1_v,\ldots,x^m_v)$,
 $\beta=(x^{{\bf g} 1}_v,\ldots,x^{{\bf g} m}_v)$,
$\gamma=(x^{{\bf h} 1}_v,\ldots,x^{{\bf h} m}_v)$.
%We need to show that $\gamma_i=s(\beta_{-i})$ for all $i\in[m]$, or equivalently $\gamma=\beta^s$. 
By definition~\eqref{eq:xgi:def}, $\beta={\bf g}(\alpha)$ and $\gamma={\bf h}(\alpha)$.
Therefore, $\gamma={\bf g}^s(\alpha)=[{\bf g}(\alpha)]^s=\beta^s$.
In other words, the $i$th component of $\gamma$ equals $s(\beta_{-i})$, which is what we needed to show.
\end{proof}
%\pagebreak

Consider $m-1$ labellings  $(x^{{\bf g} 1},\ldots,x^{{\bf g} m})_{-i}$  for $({\bf g},i)\in I$.
Applying the polymorphism inequality to these labellings 
gives
\begin{eqnarray*}
 \sum_{s\in \supp(\omega)} \omega(s)f(s((x^{{\bf g} 1},\ldots,x^{{\bf g} m})_{-i})) & \le & \frac{1}{m-1}\sum_{j\in[m]-\{i\}}f(x^{{\bf g} j}) %\hspace{13pt} \forall \alpha\in A,i\in[m]
\end{eqnarray*}
Let us multiply this inequality by weight $\lambda_{{\bf g} i}\ge 0$ (to be defined later), and apply Proposition~\ref{prop:asjgasdg}
and the fact that  $w({\bf g},{\bf h})=\sum_{s\in \supp(\omega):{\bf h}={\bf g}^s}\omega(s)$:
\begin{eqnarray*}
\lambda_{{\bf g} i} \sum_{{\bf h}:({\bf g},{\bf h})\in E} w({\bf g},{\bf h})f(x^{{\bf h} i}) - \frac{\lambda_{{\bf g} i}}{m-1}\sum_{j\in[m]-\{i\}}f(x^{{\bf g} j}) & \le & 0 \hspace{23pt} \forall ({\bf g},i)\in I
\end{eqnarray*}
Summing these inequalities over $({\bf g},i)\in I$ gives
\begin{eqnarray}
\sum_{({\bf g},i)\in I}\left[\sum_{{\bf h}:({\bf h},{\bf g})\in E} w({\bf h},{\bf g})\lambda_{{\bf h} i}
-\sum_{j\in[m]-\{i\}}\frac{\lambda_{{\bf g} j}}{m-1}\right] f(x^{{\bf g} i}) &\le & 0
%\hspace{20pt} \forall ({\bf g},i)\in I 
\label{eq:FullInequality}
\end{eqnarray}

Plugging weights $\lambda_{{\bf g}i}$ from Lemma~\ref{lemma:lambdaWeights}(a)
into~\ref{eq:FullInequality} gives inequality \eqref{eq:cyclicPol:a}. 
This proves Lemma~\ref{th:cyclicPol}(a). 

%Similarly, we can plug $\lambda_{{\bf g}i}$
Similarly, we can plug weights $\lambda_{{\bf g}i}$ from Lemma~\ref{lemma:lambdaWeights}(b)
into~\eqref{eq:FullInequality} and get inequality \eqref{eq:cyclicPol:b}. However, we need an additional argument
in order  to establish Lemma~\ref{th:cyclicPol}(b) with this strategy.
Indeed, the vector $\lambda$ in Lemma~\ref{lemma:lambdaWeights}(b) depends on the pair $(i',i'')$;
let us denote it as $\lambda^{i'i''}$. We need to show that these vectors % $\lambda^{i'i''}$ in Lemma~\ref{lemma:lambdaWeights}(b) 
can be chosen in such a way that for a given ${\bf g}\in \mathbb H$,
the components $\lambda^{i'i''}_{\bf g}$
are the same for all pairs $(i',i'')$. This can be done as follows.
Take the vector $\lambda^{12}$ constructed in Lemma~\ref{lemma:lambdaWeights}(b). For a pair of distinct indices $(i',i'')\ne (1,2)$
select a permutation $\pi$ of $[m]$ with $\pi(i')=1$, $\pi(i'')=2$, and define vector $\lambda^{i'i''}$ via
$$
\lambda^{i'i''}_{\bf g}=\lambda^{12}_{\bf g}\qquad\forall{\bf g}\in \mathbb H
\hspace{100pt}
\lambda^{i'i''}_{{\bf g}i}=\lambda^{12}_{{\bf g}\pi(i)}\qquad\forall({\bf g},i)\in I
$$
Clearly, the vector $\lambda^{i'i''}$ satisfies conditions of Lemma~\ref{lemma:lambdaWeights}(b) for the pair $(i',i'')$.
Thus, Lemma~\ref{lemma:lambdaWeights}(b) indeed implies Lemma~\ref{th:cyclicPol}(b). 

%%%%%%%%%%%%%%%%%%%%%%%%%%%%%%%%%%%%%%%%%%%%%%%%%%%%%%%%%%%%%%%%%

%%%%%%%%%%%%%%%%%%%%%%%%%%%%%%%%%%%%%%%%%%%%%%%%%%%%%%%%%%%%%%%%%%%%%%
%%%%%%%%%%%%%%%%%%%%%%%%%%%%%%%%%%%%%%%%%%%%%%%%%%%%%%%%%%%%%%%%%%%%%%
\subsection{Proof of Lemma~\ref{th:permutationInvariance}}\label{sec:proof:finiteValued:5}
In this subsection we prove the following. 
\renewcommand{\thelemmaRESTATED}{\ref{th:permutationInvariance}}
\begin{lemmaRESTATED}[restated]
Let  ${\bf g}^\ast$ be a mapping in $\mathbb G^\ast$ and
 ${\bf p}\in\calO^{(m\rightarrow m)}$ be any mapping such
that   ${\bf p}(\alpha)$ is a permutation of $\alpha$ for all $\alpha\in D^m$.
Denote
$$Range_n({\bf g}^\ast)=\{{\bf g}^\ast(x^1,\ldots,x^m)\:|\:x^1,\ldots,x^m\in D^n\}$$
For any function $f\in\Gamma$ of arity $n$ and any $(x^1,\ldots,x^m)\in Range_n({\bf g}^\ast)$ it holds that $f^m(x^1,\ldots,x^m)=f^m({\bf p}(x^1,\ldots,x^m))$.
%For any $(x^1,\ldots,x^m)\in Range_n({\bf g}^\ast)$ it holds that $f^m(x^1,\ldots,x^m)=f^m({\bf p}(x^1,\ldots,x^m))$.
%\label{th:permutationInvariance}
\end{lemmaRESTATED}

Fix function $f\in\Gamma$ of arity $n$, and 
let $\mathbb H\in{\sf Sinks}({\mathbb G,E})$ be the strongly connected component that contains ${\bf g}^\ast$.
Let $\lambda\in\mathbb R^{\mathbb H}_{\ge 0}$ be a vector constructed in Lemma~\ref{th:cyclicPol}(b).
We denote $F^\lambda(x^1,\ldots,x^m)=F^\lambda_{i}(x^1,\ldots,x^m)$ for $i\in[m]$.
\begin{lemmaX}
The following transformation does not change $F^\lambda(x^1,\ldots,x^m)$:
pick a coordinate $v\in [n]$ and permute the labels 
$(x^1_v,\ldots,x^m_v)$.
\label{lemma:HGKAJSDHAGAS}
\end{lemmaX}
\begin{proof}
It suffices to prove the claim for a permutation $\pi$ which swaps
the labels $x^i_v$ and $x^j_v$ for $i,j\in[m]$ (since any other permutation
can be obtained by repeatedly applying such swaps).
%. We claim that the cost in~\eqref{eq:pairwiseEq} does not change.
%Indeed, 
Since $m\ge 3$ there exists an index $k\in[m]-\{i,j\}$. 
We claim that for any ${\bf g}=(g_1,\ldots,g_m)\in\mathbb H$,
the
labelling $x^{{\bf g}k}=g_k(x^1,\ldots,x^m)$ is not affected by the swap above.
Indeed, it suffices to check this for coordinate $v$ (for other coordinates the claim is trivial).
Denoting the new labellings as $\tilde x^i$ and $\tilde x^{{\bf g}k}$, we can write
\begin{eqnarray*}
\tilde x^{{\bf g}k}_v=g_k(\tilde x^1_v,\ldots,\tilde x^m_v)
&=&g_{\pi(k)}(\tilde x^1_v,\ldots,\tilde x^m_v) \qquad\qquad \mbox{// since $\pi(k)=k$} \\
&=&g^\pi_{k}(\tilde x^1_v,\ldots,\tilde x^m_v) \qquad\qquad \hspace{10pt} \mbox{// by Proposition \ref{prop:cluster:basic}} \\
&=&g_{k}(\tilde x^{\pi(1)}_v,\ldots,\tilde x^{\pi(m)}_v)
=g_k(x^1_v,\ldots,x^m_v)=x^{{\bf g}k}_v
\end{eqnarray*}
Since the labellings $x^{{\bf g}k}$ do not change,
the value of $F^\lambda_k(x^1,\ldots,x^m)$ is also not affected by the swap (see its definition in eq.~\eqref{eq:Fi:def}.)
The lemma is proved.
\end{proof}

\begin{lemmaX}\label{lemma:g:retraction}
If $(x^1,\ldots,x^m)\in Range_n({\bf g}^\ast)$ then 
$(x^1,\ldots,x^m)=(x^{{\bf g}1},\ldots,x^{{\bf g}m})$ for some ${\bf g}\in \mathbb H$.
\end{lemmaX}
\begin{proof}
It suffices to show that there exists ${\bf g}\in\mathbb H$ with ${\bf g}\circ{\bf g}^\ast={\bf g}^\ast$.

Note that $\mathds{1}^{s_1\ldots s_k}\circ {\bf h}={\bf h}^{s_1\ldots s_k}$
for any $s_1,\ldots,s_k\in \supp(\omega)$ and ${\bf h}\in\calO^{(m\rightarrow m)}$, since
for any $\alpha\in D^m$ we have
%Indeed, for each $\alpha\in D^m$ we have
$$
[ \mathds{1}^{s_1\ldots s_{k}}\circ{\bf h}](\alpha)
= \mathds{1}^{s_1\ldots s_{k}}({\bf h}(\alpha))
=  [\mathds{1}({\bf h}(\alpha))]^{s_1\ldots s_{k}}
=  [{\bf h}(\alpha)]^{s_1\ldots s_{k}}
=  {\bf h}^{s_1\ldots s_{k}}(\alpha).
%\quad\qquad\forall \alpha\in D^m
$$
Therefore, conditions ${\bf g}\in \mathbb G$, ${\bf h}\in\mathbb H$ imply that ${\bf g}\circ{\bf h}\in\mathbb H$
(since ${\bf g}$ can be written as ${\bf g}=\mathds{1}^{s_1\ldots s_k}$ and there are no edges leaving $\mathbb H$).

Since $\mathbb H$ is strongly connected, there is a path in $(\mathbb G,E)$ from ${\bf g}^\ast\circ{\bf g}^\ast\in\mathbb H$ to ${\bf g}^\ast\in\mathbb H$,
i.e.\ $[{\bf g}^\ast\circ{\bf g}^\ast]^{s_1\ldots s_k}={\bf g}^\ast$ for some
$s_1,\ldots,s_k\in \supp(\omega)$.
Equivalently, ${\bf h}\circ{\bf g}^\ast\circ{\bf g}^\ast={\bf g}^\ast$ where ${\bf h}=\mathds{1}^{s_1\ldots s_k}$.
It can be checked that mapping ${\bf g}={\bf h}\circ{\bf g}^\ast$ has the desired properties.
\end{proof}

\begin{lemmaX}
If $(x^1,\ldots,x^m)\in Range_n({\bf g}^\ast)$ then $f^m(x^1,\ldots,x^m)=F^\lambda(x^1,\ldots,x^m)$.
\label{lemma:RangeFASFNA}
\end{lemmaX}
\begin{proof}
From Theorem~\ref{th:cyclicPol}(a) and Lemma~\ref{lemma:g:retraction} we get that $f^m(x^{{\bf g}1},\ldots,x^{{\bf g}m})=f^m(x^1,\ldots,x^m)$ 
for all ${\bf g}\in\mathbb H$.
Using this fact and the definition of $F^\lambda_i(\cdot)$, we can write
\begin{eqnarray*}
F^\lambda(x^1,\ldots,x^m)
&=&\frac{1}{m}\sum_{i\in[m]}F^\lambda_i(x^1,\ldots,x^m)
=\frac{1}{m}\sum_{{\bf g}\in\mathbb H}\lambda_{\bf g}\sum_{i\in[m]}f(x^{{\bf g}i})\\
&=&\sum_{{\bf g}\in\mathbb H}\lambda_{\bf g}f^m(x^{{\bf g}1},\ldots,x^{{\bf g}m})
=\sum_{{\bf g}\in\mathbb H}\lambda_{\bf g}f^m(x^1,\ldots,x^m)
=f^m(x^1,\ldots,x^m)
\end{eqnarray*}
The lemma follows.
\end{proof}

We can finally establish Lemma~\ref{th:permutationInvariance}.
For labelings $(x^1,\ldots,x^m)\in Range_n({\bf g}^\ast)$ we can write
%To establish Lemma~\ref{th:permutationInvariance}, it remains to prove that condition
%$(x^1,\ldots,x^m)\in Range_n({\bf g}^\ast)$ implies ${\bf p}(x^1,\ldots,x^m)\in Range_n({\bf g}^\ast)$
%(so that we can apply Lemma~\ref{lemma:RangeFASFNA} to $(x^1,\ldots,x^m)$ and to ${\bf p}(x^1,\ldots,x^m)$).
%This proof follows mechanically from Proposition \ref{prop:cluster:basic}, and is omitted.
$$
f^m(x^1,\ldots,x^m) \stackrel{\mbox{\tiny(1)}}=
F^\lambda(x^1,\ldots,x^m) \stackrel{\mbox{\tiny(2)}}=
F^\lambda({\bf p}(x^1,\ldots,x^m)) \stackrel{\mbox{\tiny(3)}}=
f^m({\bf p}(x^1,\ldots,x^m))
$$
where equalities (1) and (3) follow from Lemma~\ref{lemma:RangeFASFNA},
and (2) follows from Lemma~\ref{lemma:HGKAJSDHAGAS}.
Note, to be able to apply Lemma~\ref{lemma:RangeFASFNA} in (3), we
need the condition ${\bf p}(x^1,\ldots,x^m)\in Range_n({\bf g}^\ast)$.
The proof of this condition follows mechanically from the assumption
$(x^1,\ldots,x^m)\in Range_n({\bf g}^\ast)$ and Proposition \ref{prop:cluster:basic},
and is omitted.

%%%%%%%%%%%%%%%%%%%%%%%%%%%%%%%%%%%%%%%%%%%%%%%%%%%%%%%%%%%%%%%%%%%%%%%%%%%%%%%%%%%%%%%%%%%%%
%%%%%%%%%%%%%%%%%%%%%%%%%%%%%%%%%%%%%%%%%%%%%%%%%%%%%%%%%%%%%%%%%%%%%%%%%%%%%%%%%%%%%%%%%%%%%
%%%%%%%%%%%%%%%%%%%%%%%%%%%%%%%%%%%%%%%%%%%%%%%%%%%%%%%%%%%%%%%%%%%%%%%%%%%%%%%%%%%%%%%%%%%%%

\section*{Acknowledgements} We thank Andrei Krokhin for helpful discussions
and for communicating the result of Raghavendra~\cite{Raghavendra}. We also
thank the anonymous referees for their diligent work on improving the
presentation of the paper.
%

%\ifARX
%\bibliographystyle{plain}
%%%\bibliographystyle{elsarticle-num-url-modified}
%\else
%\bibliographystyle{siam}
%\fi
%\bibliography{ktz13arxiv-power-v3}

\newcommand{\noopsort}[1]{}

\appendix

\section{STP Multimorphisms Imply Submodularity}\label{sec:STP}

In this section, we consider symmetric tournament pair (STP)
multimorphisms~\cite{Cohen08:Generalising} mentioned in Section~\ref{subsec:examples}.

\begin{definitionX}
(a) A pair of operations $\langle\sqcap,\sqcup\rangle$ with $\sqcap,\sqcup:D^2\rightarrow D$
is called a \emph{symmetric tournament pair} (STP) if
\begin{subequations}
\begin{eqnarray}
%&a\sqcap a=a\sqcup a=a&\hspace{40pt}\forall a\in D\hspace{29pt}\mbox{(idempotency)} \\
&a\sqcap b=b\sqcap a,\;\; a\sqcup b=b\sqcup a&\hspace{40pt}\forall a,b\in D\hspace{20pt}\mbox{(commutativity)} \\
&\{a\sqcap b,a\sqcup b\}=\{a,b\}&\hspace{40pt}\forall a,b\in D\hspace{20pt}\mbox{(conservativity)}
\end{eqnarray}
\end{subequations}
(b) Pair $\langle\sqcap,\sqcup\rangle$ is called a \emph{submodularity operation} if there exists a total order on $D$ for which  $a\sqcap b=\min\{a,b\}$, $a\sqcup b=\max\{a,b\}$ for all $a,b\in D$. \\% where $\min,\max$ are taken with respect to some total order on $D$. \\
(c) Language $\Gamma$ admits $\langle\sqcap,\sqcup\rangle$ (or $\langle\sqcap,\sqcup\rangle$ is a multimorphism of $\Gamma$) if every function $f\in\Gamma$ of arity $n$ satisfies
\begin{equation*}
f(x\sqcap y)+f(x\sqcup y)\le f(x)+f(y)\qquad\quad\forall x,y\in D^n
\end{equation*}
\end{definitionX}
It has been shown in~\cite{Cohen08:Generalising} that if $\Gamma$ admits an STP multimorphism then $\VCSP(\Gamma)$
can be solved in polynomial time. STP multimorphisms also appeared in the dichotomy result of~\cite{kz13:jacm}:
\begin{theoremX}
Suppose a finite-valued language $\Gamma$ is \emph{conservative}, i.e.\ it contains
all possible unary cost functions $u:D\rightarrow\{0,1\}$. Then $\Gamma$ either admits
an STP multimorphism or it is NP-hard.
\end{theoremX}

In this paper we prove the following.

\begin{theoremX}
If a finite-valued language $\Gamma$ admits an STP multimorphism then it also admits a submodularity multimorphism.
\label{th:STP}
\end{theoremX}

This fact is already known; in particular, footnote~2 in~\cite{kz13:jacm} mentions
that this result is implicitly contained in \cite{Cohen08:Generalising}, and sketches a proof strategy.
However, to our knowledge a formal proof has never appeared in the literature.
This paper fills this gap. Our proof is different from the one suggested in~\cite{kz13:jacm},
and inspired some of the proof techniques used in the main part of this paper.

\subsection{Proof of Theorem \ref{th:STP}}
Consider a directed graph $G=(D,E)$. We say that $G$ is a \emph{tournament}
if for each pair of distinct labels $a,b\in D$ exactly one of the edges $(a,b)$, $(b,a)$ belongs to $E$.
We define a one-to-one correspondence between STP multimorphisms
$\langle\sqcap,\sqcup\rangle$ and tournaments $G=(D,E)$
as follows: 
$$
(a,b)\in E \qquad\Leftrightarrow \qquad (a\sqcap b,a\sqcup b)=(a,b)\qquad\quad\forall a,b\in D,a\ne b
$$
%Note, for each distinct pair of labels $a,b\in D$ exactly one of the edges $(a,b)$, $(b,a)$ belong to $E$.
It can be seen that $\langle\sqcap,\sqcup\rangle$ is a submodularity
multimorphism if and only if the corresponding graph $G$
is acyclic.
\begin{lemmaX}
Suppose a finite-valued language $\Gamma$ admits an STP multimorphism $\langle\sqcap,\sqcup\rangle$ corresponding to a tournament $G=(D,E)$,
and suppose that $G$ has a 3-cycle: $(a,b),(b,c),(c,a)\in E$.
Let $\hat G$ be the graph obtained from $G$ by reversing the orientation of edge $(a,b)$, and let $\langle\hat\sqcap,\hat\sqcup\rangle$
be the corresponding STP multimorphism. Then $\Gamma$ admits $\langle\hat\sqcap,\hat\sqcup\rangle$.
\end{lemmaX}
\begin{proof}
Let $\langle \wedge,\vee\rangle$ be the following multimorphism: 
$$
(x\wedge y,x\vee y)=
\begin{cases}
(x,y) & \mbox{if }(x,y)\in\{(a,b),(b,a)\} \\
(x\sqcap y,x\sqcup y) & \mbox{if }(x,y)\notin\{(a,b),(b,a)\} \\
\end{cases} 
$$
First, we will prove that $\Gamma$ admits $\langle\wedge,\vee\rangle$ (step 1), and then prove that $\Gamma$ admits $\langle\hat\sqcap,\hat\sqcup\rangle$ (step 2).
We fix below function $f\in\Gamma$ of arity $n$ and labellings $x,y\in D^n$.

\noindent {\bf Step 1~~}
Let us define labellings $x',y'\in D^n$ via
$$
(x'_v,y'_v)=
\begin{cases}
(x_v,x_v\sqcap y_v) & \mbox{if }(x_v,y_v)\ne (b,a) \\
(c,c) & \mbox{if }(x_v,y_v)=(b,a) \\
\end{cases} \qquad\quad\forall v\in[n]
$$
It can be checked that the following identities hold:
\begin{subequations}
\begin{eqnarray}
x'\sqcap y=y' \hspace{16pt} &~~~~~~~~~& x\sqcup y' = x' 
\label{eq:id} 
\\
x\sqcap y'=x\wedge y  &&  x'\sqcup y=x\vee y
\label{eq:id'} 
\end{eqnarray}
\end{subequations}
Let us write multimorphism inequalities for pairs $(x',y)$ and $(x,y')$:
\begin{subequations}
\begin{eqnarray}
\underline{\underline{f(x'\sqcap y)}}+f(x'\sqcup y)&\le& \underline{f(x')}+f(y) \label{eq:A'} \\
f(x\sqcap y')+\underline{f(x\sqcup y')}&\le& f(x)+\underline{\underline{f(y')}}\label{eq:A''} 
\end{eqnarray}
\end{subequations}
Summing \eqref{eq:A'} and \eqref{eq:A''}, cancelling terms using \eqref{eq:id}, and then substituting expressions using~\eqref{eq:id'} gives
\begin{eqnarray}
f(x\wedge y)+f(x\vee y)&\le& f(x)+f(y) \label{eq:A'''} 
\end{eqnarray}
%Using~\eqref{eq:id'}, we conclude that $\langle\sqcap',\sqcup'\rangle$ is a multimorphism of $\Gamma$.

\noindent {\bf Step 2~~}
Let us define labellings $x',y'\in D^n$ via
$$
(x'_v,y'_v)=
\begin{cases}
(x_v\wedge y_v,y_v) & \mbox{if }(x_v,y_v)\ne (a,b) \\
(c,c) & \mbox{if }(x_v,y_v)=(a,b) \\
\end{cases} \qquad\quad\forall v\in[n]
$$
It can be checked that the following identities hold:
\begin{subequations}
\begin{eqnarray}
 x'\vee y\,= y' \hspace{15pt} &~~~~~~~~~& x\wedge y' =  x' 
\label{eq:Bid} 
\\
x\vee y'=x\, \hat\sqcup\, y &&  x'\wedge y=x\,\hat\sqcap\, y
\label{eq:Bid'} 
\end{eqnarray}
\end{subequations}
Let us write multimorphism inequalities for pairs $(x',y)$ and $(x, y')$:
\begin{subequations}
\begin{eqnarray}
f( x'\wedge y)+\underline{\underline{f( x'\vee y)}}&\le& \underline{f( x')}+f(y) \label{eq:B'} \\
\underline{f(x\wedge  y')}+f(x\vee y')&\le& f(x)+\underline{\underline{f( y')}}\label{eq:B''} 
\end{eqnarray}
\end{subequations}
Summing \eqref{eq:B'} and \eqref{eq:B''}, cancelling terms using \eqref{eq:Bid}, and then substituting expressions using~\eqref{eq:Bid'} gives
\begin{eqnarray}
f( x\,\hat\sqcap\, y)+f(x\,\hat\sqcup\,  y)&\le& f(x)+f(y) \label{eq:B'''} 
\end{eqnarray}
%Using~\eqref{eq:Bid'}, we conclude that $\langle\hat\sqcap,\hat\sqcup\rangle$ is a multimorphism of $\Gamma$.
\end{proof}

We call the operation of reversing the orientation of edge $(a,b)\in E$ in a graph $G=(D,E)$
a {\em valid flip} if $(a,b)$ belongs to a $3$-cycle.
To prove Theorem~\ref{th:STP}, it thus suffices to show the following:
\begin{itemize}
\item {\em For any tournament $G$ there exists a sequence of valid flips that makes it acyclic.}
\end{itemize}
Such sequence can be constructed as follows: (1) start with a subset $B\subseteq D$
with $|B|=3$; (2) perform valid flips in $G[B]$ to make it acyclic, where $G[B]=(B,E[B])$ is the subgraph of $G$
induced by $B$; (3) if $B\ne D$, add a vertex $c\in D-B$ to $B$ and repeat step 2.
The lemma below shows how to implement step 2.
\begin{lemmaX}
Suppose that $G=(B',E)$ is a tournament, $B'=B\cup\{c\}$ with $c\notin B$ and subgraph $G[B]$ is acyclic.
Then there exists a sequence of valid flips that makes $G$ acyclic.
\end{lemmaX}
\begin{proof}
Suppose that $G$ has a cycle $\calC$, then it must pass through $c$ (since $G[B]$ is acyclic):
$\calC=\ldots\!\rightarrow\! b\!\rightarrow\! c\!\rightarrow\! a\!\rightarrow\! \ldots$.
Since there is a path from $a$ to $b$ in $G[B]$, we must have $(a,b)\in E$
(again, due to acyclicity of $G[B]$).
Thus, $c\!\rightarrow\! a\!\rightarrow\! b\!\rightarrow\! c$ is a 3-cycle in $G$.

Let us repeat the following procedure while possible: pick such cycle and flip edge $(c,a)$ to $(a,c)$.
This operation decreases the number of edges  in $G$ coming out of $c$. Therefore,
it must terminate after a finite number of steps and yield an acyclic graph $G$.
\end{proof}

\end{document}